\documentclass[a4paper,USenglish,cleveref, nameinlink,autoref, numberwithinsect]{lipics-v2021}
\hideLIPIcs
\usepackage[T1]{fontenc}
\def\doi#1{\href{https://doi.org/\detokenize{#1}}{\url{https://doi.org/\detokenize{#1}}}}
\usepackage{graphicx} 
\usepackage{amsmath}
\usepackage{amssymb}
\usepackage{booktabs}
\usepackage{enumerate}
\usepackage{bm}
\usepackage{nicefrac}

\usepackage[algo2e, vlined, ruled, algosection]{algorithm2e}

\newcommand{\cf}{\ensuremath{\mathcal{F}}}
\newcommand{\cl}{\ensuremath{\mathcal{L}}}

\newcommand{\B}{\ensuremath{\mathcal{B}}}
\newcommand{\R}{\ensuremath{\mathcal{R}}}
\newcommand{\I}{\ensuremath{\mathcal{I}}}
\newcommand{\A}{\ensuremath{\mathcal{A}}}
\newcommand{\Oc}{\ensuremath{\mathcal{O}}}
\newcommand{\sP}{\ensuremath{\mathsf{P}}}
\newcommand{\NP}{\ensuremath{\mathsf{NP}}}
\newcommand{\XP}{\ensuremath{\mathsf{XP}}}
\newcommand{\FPT}{\ensuremath{\mathsf{FPT}}}
\newcommand{\W}{\ensuremath{\mathsf{W[1]}}}
\newcommand{\WT}{\ensuremath{\mathsf{W[2]}}}
\newcommand{\N}{\ensuremath{\mathbb{N}}}

\newcommand{\fR}{\ensuremath{\mathfrak{R}}}
\newcommand{\D}{\ensuremath{\mathfrak{D}}}

\newcommand{\MinLeafF}{\textsc{Min-Leaf \cf-Tree}}
\newcommand{\MaxLeafF}{\textsc{Max-Leaf \cf-Tree}}

\newcommand{\MinLeafOhne}{\textsc{Min-Leaf}}
\newcommand{\MaxLeafOhne}{\textsc{Max-Leaf}}
\newcommand{\MinInOhne}{\textsc{Min-Internal}}
\newcommand{\MaxInOhne}{\textsc{Max-Internal}}

\newcommand{\MaxLeaf}{\textsc{Max-Leaf Spanning Tree}}

\newcommand{\MinInF}{\textsc{Min-Internal \cf-Tree}}
\newcommand{\MaxInF}{\textsc{Max-Internal \cf-Tree}}

\newcommand{\ConDom}{\textsc{Connected Dominating Set}}
\newcommand{\WCS}{\textsc{Weighted Circuit Satisfiability}}

\newcommand{\T}{\ensuremath{\mathfrak{T}}}
\newcommand{\Z}{\ensuremath{Z^*}}
\newcommand{\true}{\texttt{true}}
\newcommand{\false}{\texttt{false}}
\newcommand{\z}{\ensuremath{\overset{\scriptscriptstyle{Z*}}{\longrightarrow}}}

\newcommand{\BIP}{\textsc{One-Sided Grundy Total Domination}}

\newcommand{\ZFS}{\textsc{Zero Forcing Set}}

\usepackage{tikz}
\usetikzlibrary{arrows}
\usetikzlibrary{math}
\usetikzlibrary{calc}
\tikzstyle{treeedge}=[-stealth', ultra thick]
\tikzstyle{normaledge}=[black]
\tikzset{vertex/.style = {draw, circle, thick,minimum size = 5pt,inner sep=0pt}}
\tikzset{vertexKlein/.style = {minimum size = 0pt,inner sep=0pt}}

\crefname{observation}{Observation}{Observations}
\Crefname{observation}{Observation}{Observations}

\title{Breadth-First Search Trees with Many or Few Leaves}

\author{Jesse Beisegel}{Institute of Mathematics, Brandenburg University of Technology, Cottbus, Germany}{jesse.beisegel@b-tu.de}{https://orcid.org/0000-0002-8760-0169}{}
\author{Ekkehard Köhler}{Institute of Mathematics, Brandenburg University of Technology, Cottbus, Germany}{ekkehard.koehler@b-tu.de}{}{}
\author{Robert Scheffler}{Institute of Mathematics, Brandenburg University of Technology, Cottbus, Germany}{robert.scheffler@b-tu.de}{https://orcid.org/0000-0001-6007-4202}{}
\author{Martin Strehler}{Department of Mathematics, Westsächsische Hochschule Zwickau, Zwickau, Germany}{martin.strehler@fh-zwickau.de}{https://orcid.org/0000-0003-4241-6584}{}

\authorrunning{J. Beisegel, E. Köhler, R. Scheffler, and M. Strehler} 




\nolinenumbers

\begin{document}

\maketitle              
\begin{abstract}
The Maximum (Minimum) Leaf Spanning Tree problem asks for a spanning tree with the largest (smallest) number of leaves. As spanning trees are often computed using graph search algorithms, it is natural to restrict this problem to the set of search trees of some particular graph search, e.g., find the Breadth-First Search (BFS) tree with the largest number of leaves. We study this problem for Generic Search (GS), BFS  and Lexicographic Breadth-First Search (LBFS) using search trees that connect each vertex to its first neighbor in the search order (\emph{first-in trees}) just like the classic BFS tree. In particular, we analyze the complexity of these problems, both in the classical and in the parameterized sense. 

Among other results, we show that the minimum and maximum leaf problems are in \FPT{} for the first-in trees of GS, BFS and LBFS when parameterized by the number of leaves in the tree. However, when these problems are parameterized by the number of internal vertices of the tree, they are \W-hard for the first-in trees of GS, BFS and LBFS.

\end{abstract}

\section{Introduction}

The study of the number of leaves in a spanning tree is a well-established area within graph theory. On the one hand, the problem of minimizing the number of leaves is related to the \textsc{Hamiltonian Path} problem: a spanning tree of a graph $G$ with exactly one leaf corresponds to a Hamiltonian path in $G$.\footnote{\label{fn:leaf}Note that we only consider rooted spanning trees here and we do not count the root as a leaf even if it has degree~1, following the convention in~\cite{bergougnoux2025parameterized}.} More broadly, this problem is equivalently known as the \textsc{Maximum Internal Spanning Tree} problem (MIST). For this problem, there are both Fixed-Parameter Tractable (\FPT) results when parameterized by the number of leaves ~\cite{fomin2013linear,li2017deeper,prieto2003either} as well as polynomial-time algorithms for specific graph classes. Linear-time algorithms exist for, e.g., interval graphs, block graphs, cactus graphs, or chain graphs~\cite{li2022simple,sharma2022algorithms}.

On the other hand, maximizing the number of leaves is known as the 
\textsc{Maximum Leaf Spanning Tree} problem (MLST). This particular problem is NP-hard not only for general graphs (problem ND2 in \cite{GareyJohnson}), but also for specific graph classes such as planar cubic graphs \cite{reich16complexity}. When viewed from the perspective of parameterized complexity, there is a wide range of \FPT{} algorithms for this problem when parameterized by the number of leaves~\cite{bodlaender1993linear,
bonsma2008spanning,fellows1992well,fellows2000coordinatized,kneis2011new,raible2010amortized,zehavi2018k-leaf}. The parametric dual of the problem, i.e., finding a spanning tree with minimal number of internal vertices, is equivalent to the \WT-complete problem of finding a minimum connected dominating set~\cite{fujie2003exact}.

Another line of research has focused on studying the properties of search orderings and their associated search trees (see \cite{corneil2008unified} for a survey on graph searches). Graph searches\footnote{Note that in this paper, the term \emph{graph searches} always refers to \emph{connected searches}, i.e., each prefix of a search ordering induces a connected graph.} like \emph{Breadth-First Search} (BFS), \emph{Depth-First Search} (DFS), or their variants, are crucial components in a multitude of graph algorithms. For example, \emph{Lexicographic BFS} (LBFS) or \emph{Maximum Cardinality Search} (MCS) can be used to recognize chordal graphs in linear time~\cite{rosetarjanlueker76,tarjan1984linear}. Furthermore, the search ordering derived from LBFS or MCS on a chordal graph can also be used to compute an optimal coloring, a maximum independent set or a maximum clique in linear time. Here, the end vertex of the search plays a pivotal role, as it is proven to be simplicial and a valid starting vertex for a perfect elimination ordering (see~\cite{golumbic1980algorithmic} for details). As a consequence, the end vertex problem has become a key area of study within this context, see, e.g.,~\cite{beisegel2019end,charbit2014influence,corneil2010end,kratsch2015end,rong2022graph,rong2026linear}.

Another structure related to graph searches are their search trees. The complexity of deciding whether a spanning tree can be constructed by a certain graph search was first studied in the 1980s for BFS and DFS~\cite{hagerup1985biconnected,hagerup1985recognition,korach1989dfs,manber1990recognizing}. More recently, this problem has also been considered for searches such as LBFS and MCS~\cite{beisegel2021recognition,beisegel2024graph}. As observed in~\cite{scheffler2025partial}, this research unveiled a strong relation to the end vertex problem. This relation also lead to the study of leaf recognition, i.e., the question whether a vertex can be a leaf in a search tree~\cite{scheffler2025leaves}.

Recently, Bergougnoux et al.~\cite{bergougnoux2025parameterized} have combined the study of spanning trees with few or many leaves with the concept of search trees.
They examine the algorithmic problems of finding DFS trees with either few or many leaves, highlighting potential applications in network design. They showed that the problems are hard when parameterized by the numbers of leaves but in \FPT{} when parameterized by the number of internal vertices. It therefore seems logical to explore similar questions for other searches and their corresponding trees. 

Note that BFS trees and DFS trees differ significantly in their structure. While DFS trees connect each vertex $v$ to the last neighbor visited before~$v$, BFS trees connect each vertex to its first visited neighbor. Following~\cite{beisegel2021recognition}, we call the first concept the \emph{last-in tree} (or $\cl$-tree) of the search ordering, whereas the second is called the \emph{first-in tree} (or $\cf$-tree).  DFS trees tend to have few leaves as, for example, every Hamiltonian path is a DFS tree. In contrast, BFS trees tend to have many leaves. For example, if we start the BFS in a universal vertex, then all vertices except the root are leaves. It therefore seems to be a reasonable hypothesis that the parameterized complexity of finding BFS trees with many or few leaves differs significantly from that of DFS trees for the natural parameters \emph{number of leaves} and \emph{number of internal vertices}.

\subparagraph{Our results.}
We study the problem of finding $\cf$-trees with few or many leaves. For a given graph $G$ and a search paradigm $\A$ such as BFS, we call an ordering constructed by $\A$ when applied to $G$ an \emph{$\A$-ordering} of $G$. We consider both the minimizing and maximizing versions of the following two problems, which differ in their use of parameter $k$. The first is the \textsc{Min-Leaf (Max-Leaf) \cf-Tree} of search $\A$, where given a graph $G$ and a parameter $k$ we ask for an $\A$-ordering whose $\cf$-tree has at most (least) $k$~leaves.

The second is the \textsc{Min-Internal (Max-Internal) \cf-Tree} of search $\A$, where given a graph $G$ and a parameter $k$ we ask for an $\A$-ordering of $G$ whose $\cf$-tree has at most (least) $k$ internal vertices. Note that even though the parameter $k$ does not appear in the problem name, it is essential in differentiating these two  problems, as -- without parameterization -- minimizing the number of leaves is the same as maximizing the number of the internal vertices and vice versa.

In particular, we consider GS, BFS, LBFS and related graph searches that share the so-called clique-starter property. Our results presented in this paper are summarized in \cref{tab:results-tree}. We are not aware of any previous work on these problems for these searches.

\begin{table}[t]
	\centering
	\caption{Results given in this paper. The first part of an entry concerns hardness results and the second entry concerns algorithms. For all considered cases, \W- and \WT-hardness also implies \NP-hardness of the non-parameterized problem.}\label{tab:results-tree}
    \resizebox{\textwidth}{!}{
	\begin{tabular}{l c c c c c c c c}
		\addlinespace
		\toprule
		results   &~~~& GS &~~~& BFS                 &~~~& LBFS           &~~~& other clique starters     \\[0.8ex] \toprule
		\MinLeafOhne{}     & & 
		  NP-compl. / FPT & & 
		NP-compl. / FPT & &
		NP-compl. / FPT & & 
		NP-hard / ?     
		   \\[0.8ex]
		\MaxLeafOhne{}  & & 
		NP-compl. / FPT & & 
		NP-compl. / FPT       & & 
		NP-compl. / FPT       & & 
		NP-hard / ?  
		\\[0.8ex]
        \MaxInOhne{}      & &
		  W[1]-compl. / XP    & & 
		W[1]-hard / XP & &
		W[1]-hard / ? & &
		W[1]-hard / ?        
		 \\[0.8ex]
		\MinInOhne{} & &
		W[2]-compl. / XP & & 
		W[2]-hard / XP & & 
		para-NP-hard & & 
		W[2]-hard / ? 
        \\[0.8ex] 
		\bottomrule \addlinespace
	\end{tabular}
    }
\end{table}

\section{Preliminaries}\label{sec:prelim}

All graphs in this paper are undirected, simple, connected, and non-empty. Given a graph~$G$, we denote the set of \emph{vertices} by $V(G)$ and the set of \emph{edges} by $E(G)$. As is customary, we use $n=|V|$ and $m=|E|$ to represent the number of vertices and edges, respectively. We use $N_G(v)$ to denote the \emph{(open) neighborhood} of a vertex $v$ in the graph $G$, while $N_G[v] = N_G(v) \cup \{v\}$ is the \emph{closed neighborhood} For further graph theoretic concepts, we refer to~\cite{west2001introduction}. 

A \emph{graph search} $\cal A$ is a procedure to traverse all vertices of a connected graph. The associated \emph{search ordering} $\sigma=(v_1,\dots,v_n)$ represents the order in which the vertices are visited. The \emph{bandwidth of a vertex ordering} $\sigma$ is the maximum distance of adjacent vertices in the ordering, i.e., $\operatorname{bw}(\sigma)=\max_{v_iv_j\in E(G)} |i-j|$. The \emph{bandwidth of a graph} $\operatorname{bw}(G)$ is the minimum bandwidth over all orderings.

In this paper, we consider the following searches $\cal A$. \emph{Generic Search} (GS) imposes no other conditions on the search than connectivity, meaning that after the first node is visited each newly visited node must be a neighbor of an already visited node. BFS refers to the standard \emph{Breadth-First Search} implemented via a queue data structure~\cite{Kleinberg06algorithmdesign}. 
\emph{Lexicographic Breadth-First Search} (LBFS) is similar to BFS, but uses partition refinement as an improved tie-breaker~(see \cref{lbfs} and \cite{habib2000lexbfs,rosetarjanlueker76} for details). 

\begin{algorithm2e}
	\KwIn{Connected graph $G=(V,E)$ and a distinguished vertex $ s \in V $}
	\KwOut{A vertex ordering $ \sigma $}
	\Begin{
		$ label(s) \leftarrow (n) $\;
		
		\lForEach{vertex $v \in V\setminus\{s\}$}{assign to $ v $ the empty label}
		
		\For{$ i \leftarrow 1 $ to $ n $}{pick an unnumbered vertex $ v $ with lexicographically largest label\;
			$ \sigma(i) \leftarrow v$\;
			\lForEach{unnumbered vertex $ w \in N(v) $}{append $ (n-i) $ to $ label(w) $}}		
	}\caption{Lexicographic Breadth First Search}\index{LBFS}
	\label{lbfs}
\end{algorithm2e}

An ${\cal A}^+_\rho$-search is a search following the search paradigm $\cal A$ using a given linear ordering $\rho$ of the vertices as a tie-breaker. Whenever $\cal A$ faces a tie, i.e., several vertices are valid choices for the next vertex, it uses the leftmost of these vertices in $\rho$ to proceed. Such $+$-searches usually appear in multi-sweep algorithms where $\rho$ was obtained through a preceding search.

We say that a graph search $\A$ is a \emph{clique starter} if for every graph $G$, every clique $C$ of $G$ and every ordering $\sigma$ of the vertices in $C$, there is an $\A$-ordering of $G$ that starts with $\sigma$. Most of the graph searches considered in the literature, in particular GS, BFS, and LBFS, are clique starters.

A spanning tree of $G$ rooted at a vertex $r\in V(G)$ is denoted by $(T,r)$ or simply by $T$ whenever the root is clear from context. A vertex $v \in V(T)$ is a \emph{leaf} if it has no descendants. Otherwise, it is called an \emph{internal vertex}. Note that we do not consider the root to be a leaf, even if it has degree one. A spanning tree $T$ of $G$ is an \emph{\cf-tree} associated with search ordering $\sigma=(r=v_1,v_2,\dots,v_n)$ if it is rooted in $r$ and for any vertex $v_i$, $1<i\le n$, among all vertices in $N(v_i)$, the parent of $v_i$ in the tree $T$ is leftmost in $\sigma$. 

A graph is a \emph{split graph} if its vertex set can be partitioned into a clique and an independent set. A graph $G$ is \emph{weakly chordal} if neither $G$ nor its complement $\overline{G}$ contain an induced cycle of length greater than four. The following two lemmas present graph operations that result in a weakly chordal graph if and only if the original graph is weakly chordal.

\begin{lemma}[Spinrad and Sritharan~\cite{spinrad1995algorithms}]\label{lemma:2-pair}
Let $G$ be a graph with a two pair $\{u, w\}$, i.e., a pair of vertices such that every induced $u$-$w$-path in $G$ has length~2. Then $G$ is weakly chordal if and only if the graph that is constructed by adding the edge~$uw$ to~$G$ is weakly chordal.
\end{lemma}

\begin{lemma}[Beisegel et al.~\cite{beisegel2021recognition}]\label{lemma:simplicial}
Let $G$ be a graph and $v \in V(G)$ such that $v$ is simplicial ($N_G(v)$ induces a clique in $G$) or adjacent to at least $n(G) - 2$ vertices of $G$. Then $G$ is weakly chordal if and only if $G - v$ is weakly chordal.
\end{lemma}

For a detailed presentation of the concepts of parameterized complexity and the definition of the considered width-parameters, we refer the reader to \cite{cygan2015param,downey2013fundamentals}. 

\section{Number of Leaves as Parameter}\label{sec:leaves}

Here, we will present parameterized algorithms for both \MinLeafF{} and \MaxLeafF{} for several search paradigms. Note that deciding whether a particular vertex can be a leaf of an \cf-tree is \NP-complete for almost all considered searches~\cite{scheffler2025leaves}. This also holds for sets of $k$ fixed vertices with $k \geq 2$ by just appending leaves to the graph. Therefore, this problem seems to require a more sophisticated approach than simply checking all possible sets of $k$ vertices. 

To see the difference between the general spanning tree problems and the \cf-tree problems, we compare them on the examples given in~\cref{fig:path-triangles-star-ladders}. By definition, any \cf-tree is also a spanning tree. Therefore, when considering the minimization problems, any minimum leaf number for spanning trees is a lower bound for the number of leaves of an \cf-tree. However, the solutions can be arbitrarily far apart from each other: In the left graph of~\cref{fig:path-triangles-star-ladders}, we give a family of graphs where the minimum number of leaves in a general spanning tree is~$1$ and the minimum number of leaves in a BFS-tree is $\operatorname{\Omega}(n)$. For the maximization problems, any solution of the \cf-tree problem is clearly a lower bound for the spanning tree problem. However, the right graph of~\cref{fig:path-triangles-star-ladders} shows an example where the maximum number of leaves in a BFS-tree is $\Oc(\sqrt{n})$, while there are spanning trees with $\operatorname{\Omega}(n)$ leaves. We leave open whether this gap can be increased, i.e., whether there is a family of graphs in which the maximum number of leaves in any \cf-tree is constant and there are spanning trees with $\operatorname{\Omega}(n)$ leaves.

\begin{figure}
\begin{minipage}{0.38\linewidth}
\centering
  \begin{tikzpicture}[scale=0.4]
\footnotesize
\node[vertex] (1) at (0,0) {};
\node[vertex] (2) at (1,1.5) {};
\node[vertex] (3) at (2,0) {};
\node[vertex] (4) at (3,1.5) {};
\node[vertex] (5) at (4,0) {};
\node (6) at (5,0) {};
\node (7) at (5,1.5) {};
\draw[normaledge] (5) -- (6);
\draw[normaledge] (5) -- (7);
\node (88) at (6,0) {};
\node (89) at (6,1.5) {};
\node[vertex] (9) at (7,0) {};
\draw[normaledge] (88) -- (9);
\draw[normaledge] (89) -- (9);
\node[vertex] (10) at (8,1.5) {};
\node[vertex] (11) at (9,0) {};
\node[align = center] at ($(5)!.5!(9)$) {\ldots};

\draw[normaledge] (1) -- (2);
\draw[normaledge] (1) -- (3);
\draw[normaledge] (2) -- (3);
\draw[normaledge] (3) -- (4);
\draw[normaledge] (3) -- (5);
\draw[normaledge] (4) -- (5);
\draw[normaledge] (9) -- (10);
\draw[normaledge] (9) -- (11);
\draw[normaledge] (10) -- (11);

\end{tikzpicture}

\end{minipage}\hfill
\begin{minipage}{0.6\linewidth}
\centering
  \begin{tikzpicture}[scale=0.23]
  \footnotesize

  \node[vertex] (1) at (2,-1) {};

  \node[vertex] (a1) at (-2,-2) {};
  \node[vertex] (a2) at (-2,-4) {};
  \node[vertex] (a3) at (-4,-2) {};
  \node[vertex] (a4) at (-4,-4) {};
  \node[] (a11) at (-6,-3) {};
  \node[vertexKlein] (a5) at (-6,-2) {};
  \node[vertexKlein] (a6) at (-6,-4) {};
  \node[vertexKlein] (a55) at (-8,-2) {};
  \node[vertexKlein] (a66) at (-8,-4) {};
  \node[vertex] (a7) at (-10,-2) {};
  \node[vertex] (a8) at (-10,-4) {};
  \node[] (a12) at (-8,-3) {};
  \node[vertex] (a9) at (-12,-2) {};
  \node[vertex] (a10) at (-12,-4) {};
\node[align = center] at ($(a11)!.5!(a12)$) {\ldots};

  \node[vertex] (b1) at (-2,2) {};
  \node[vertex] (b2) at (-2,0) {};
  \node[vertex] (b3) at (-4,2) {};
  \node[vertex] (b4) at (-4,0) {};
  \node[] (b11) at (-6,1) {};
  \node[vertexKlein] (b5) at (-6,2) {};
  \node[vertexKlein] (b6) at (-6,0) {};
  \node[vertexKlein] (b55) at (-8,2) {};
  \node[vertexKlein] (b66) at (-8,0) {};
  \node[vertex] (b7) at (-10,2) {};
  \node[vertex] (b8) at (-10,0) {};
  \node[] (b12) at (-8,1) {};
  \node[vertex] (b9) at (-12,2) {};
  \node[vertex] (b10) at (-12,0) {};
\node[align = center] at ($(b11)!.5!(b12)$) {\ldots};

  \node[vertex] (c1) at (6,-4) {};
  \node[vertex] (c2) at (6,-2) {};
  \node[vertex] (c3) at (8,-4) {};
  \node[vertex] (c4) at (8,-2) {};
  \node[] (c11) at (10,-3) {};
  \node[vertexKlein] (c5) at (10,-4) {};
  \node[vertexKlein] (c6) at (10,-2) {};
  \node[vertexKlein] (c55) at (12,-4) {};
  \node[vertexKlein] (c66) at (12,-2) {};
  \node[vertex] (c7) at (14,-4) {};
  \node[vertex] (c8) at (14,-2) {};
  \node[] (c12) at (12,-3) {};
  \node[vertex] (c9) at (16,-4) {};
  \node[vertex] (c10) at (16,-2) {};
\node[align = center] at ($(c11)!.5!(c12)$) {\ldots};

\draw[normaledge] (1) -- (a1);
\draw[normaledge] (1) -- (a2);
\draw[normaledge] (a1) -- (a2);
\draw[normaledge] (a1) -- (a3);
\draw[normaledge] (a2) -- (a4);
\draw[normaledge] (a3) -- (a4);
\draw[normaledge] (a3) -- (a5);
\draw[normaledge] (a4) -- (a6);
\draw[normaledge] (a55) -- (a7);
\draw[normaledge] (a66) -- (a8);
\draw[normaledge] (a7) -- (a8);
\draw[normaledge] (a7) -- (a9);
\draw[normaledge] (a8) -- (a10);
\draw[normaledge] (a9) -- (a10);

\draw[normaledge] (1) -- (b1);
\draw[normaledge] (1) -- (b2);
\draw[normaledge] (b1) -- (b2);
\draw[normaledge] (b1) -- (b3);
\draw[normaledge] (b2) -- (b4);
\draw[normaledge] (b3) -- (b4);
\draw[normaledge] (b3) -- (b5);
\draw[normaledge] (b4) -- (b6);
\draw[normaledge] (b55) -- (b7);
\draw[normaledge] (b66) -- (b8);
\draw[normaledge] (b7) -- (b8);
\draw[normaledge] (b7) -- (b9);
\draw[normaledge] (b8) -- (b10);
\draw[normaledge] (b9) -- (b10);

\draw[normaledge] (1) -- (c1);
\draw[normaledge] (1) -- (c2);
\draw[normaledge] (c1) -- (c2);
\draw[normaledge] (c1) -- (c3);
\draw[normaledge] (c2) -- (c4);
\draw[normaledge] (c3) -- (c4);
\draw[normaledge] (c3) -- (c5);
\draw[normaledge] (c4) -- (c6);
\draw[normaledge] (c55) -- (c7);
\draw[normaledge] (c66) -- (c8);
\draw[normaledge] (c7) -- (c8);
\draw[normaledge] (c7) -- (c9);
\draw[normaledge] (c8) -- (c10);
\draw[normaledge] (c9) -- (c10);

\begin{scope}[xshift=2cm,yshift=-1cm]

  \node[vertexKlein] (d1) at (110:3.4cm) {};
  \node[vertexKlein] (d2) at (90:3.3cm) {};

  \draw[normaledge] (1) -- (d1);
  \draw[normaledge] (1) -- (d2);

  \node[vertexKlein] (f1) at (10:3.3cm) {};
  \node[vertexKlein] (f2) at (30:3.3cm) {};

  \draw[normaledge] (1) -- (f1);
  \draw[normaledge] (1) -- (f2);

  \draw[dotted, very thick] (50:3cm) arc [start angle=50, end angle=75, radius=3cm];  
\end{scope}

\end{tikzpicture}

\end{minipage}
    \caption{In the left graph the minimum number of leaves in a spanning tree is 1, whereas any BFS-tree has at least $\nicefrac{n}{2}$ leaves.
    The right graph is a star of $k$ ladders with $2k$ vertices each. This graph has a spanning tree with $k^2$ leaves while every BFS-tree has at most $3k$ leaves.}
    \label{fig:path-triangles-star-ladders}
\end{figure}
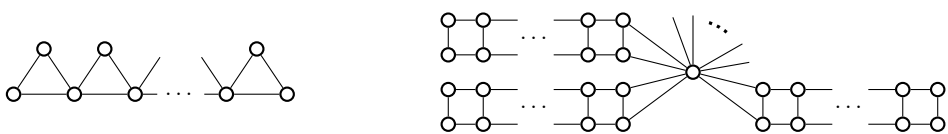

To start with, we focus on \emph{layered} searches~\cite{corneil2016tie}, i.e., graph searches that traverse the distance layers of the start vertex in increasing order. Well-known examples of such searches are BFS and LBFS. 

\begin{observation}\label{obs:layer}
    For graphs with at least two vertices, the number of leaves of the \cf-tree of some layered search ordering is at least the maximum number of vertices in a layer.
\end{observation}

It is a well-known fact that for layered searches the edges of a graph only connect vertices in the same layer or in consecutive layers. Thus, the observation above implies the following.

\begin{corollary}
    Let $G$ be a graph with at least two vertices. If the \cf-tree of a layered graph search ordering $\sigma$ of $G$ has at most $k$ leaves, then the bandwidth of $\sigma$ is at most~$2k - 1$.
\end{corollary}

We can improve this bound when we consider BFS orderings.

\begin{theorem}    
    If the \cf-tree $T$ of a BFS ordering $\sigma$ of graph $G$ has at most $k$ leaves, then the bandwidth of $\sigma$ is at most~$k$.
\end{theorem}

\begin{proof}
    Let $v$ be a vertex in $G$ and let $w$ be the rightmost neighbor of $v$ in $\sigma$. Let $\sigma' = (v_1, \dots, v_\ell = w)$ be the subordering of $\sigma$ between $v$ and $w$, i.e., $v_1$ is the direct successor of $v$ in $\sigma$. The parents of all vertices in $\sigma'$ lie to the left of $v_1$ in~$\sigma$. Thus, none of the vertices in $\sigma'$ can have a descendant in $T$ that is also in $\sigma'$. Therefore, there are at least $\ell$ leaves in $T$. Thus, $\ell \leq k$.
\end{proof}

As we will see later (\cref{thm:max-leaf-F-NPc,thm:max-int-f-w-hard}), both the \MinLeafF{} and the \MaxLeafF{} problem is \NP-hard for BFS and for LBFS. However, we can show that these problems are in \FPT.

\begin{theorem}
    \MinLeafF{} and \MaxLeafF{} of BFS and LBFS can be solved in \FPT{} time.
\end{theorem}

\begin{proof}
    We first consider \MinLeafF. For every possible start vertex $r$, we do the following. We first compute the layers of the BFS. If there is a layer with more than $k$ vertices, then we can reject that start vertex, due to \cref{obs:layer}. Otherwise, we use dynamic programming to solve the problem. Let $V_{\leq i}$ be the union of all layers with index $\leq i$. For every ordering $\sigma$ of the $i$-th layer, we have the value $M[i,\sigma]$ that contains $\infty$ if there is no (L)BFS ordering starting with $r$ that contains $\sigma$ as subordering. Otherwise, it contains the minimal number of leaves in $V_{\leq i}$ of the \cf-tree of an (L)BFS ordering of $V_{\leq i+1}$ starting in $r$ that has $\sigma$ as subordering.
    
    Layer~0 has exactly one ordering $(r)$ and we have $M[0,(r)] = 0$. So assume that for some value $i$, we have computed the values $M[i-1,\sigma]$ for all orderings $\sigma$ of layer~$(i-1)$. Now consider an entry $M[i,\tau]$. If layer~$i$ is the last layer, then all vertices of the layer are leaves in the $\cf$-tree of every BFS ordering starting in $r$. If it is not the last layer, then the ordering $\tau$ directly implies which of the vertices of layer~$i$ has a child in layer~$i+1$. If some vertex does not have a child, then it has to be a leaf in the \cf-tree of every BFS ordering starting in $r$ that contains $\tau$ as subordering. Therefore, we compute the number $\ell(i, \tau)$ of leaves in layer~$i$ for the ordering $\tau$. To compute $M[i,\tau]$, we now check for every ordering $\sigma$ of layer~$i-1$ whether $\tau$ can be the ordering of layer~$i$ when $\sigma$ is the ordering of layer~$i-1$. Note that this can be done in polynomial time. If this is the case, then the value $M[i-1,\sigma] + \ell(i,\tau)$ gives the number of leaves in $V_{\le i}$ if $\sigma$ and $\tau$ are the orderings of layer $i-1$ and layer~$i$, respectively. We minimize this value over all suitable orderings~$\sigma$. The minimum value is chosen for $M[i, \tau]$.

    If we have computed all entries of $M$, we check for the last layer $i$ whether there is an entry $M[i, \sigma] \leq k$. If so, then we return ``Yes''. If no start vertex works, then we return ``No''. Maximization works analogously with the only difference being that $M[i,\sigma]$ contains the maximum value of leaves or $0$ if there is no possible (L)BFS ordering. Note that we can stop the maximization procedure directly with a positive answer as soon as we find a layer with $\geq k$ elements or an entry $M[i, \sigma] \geq k$. The minimization procedure for a fixed start vertex can be stopped if every entry $M[i,\sigma]$ of some layer $i$ is larger than~$k$.

    Finally, let us consider the running time of our algorithm. We have to check at most $n$ start vertices. For every of the at most $n$ layers, there are at most $k!$ entries of $M$. To compute such an entry, we have to check at most $k!$ other entries and have to do a computation polynomial in $k$. So the final running time is in $\Oc(k!^2 \cdot k^{\Oc(1)} \cdot n^2)$. 
\end{proof}

Next we show that \MaxLeafF{} of GS is equivalent to the problem \MaxLeaf.

\begin{theorem}\label{thm:tree-dom-equiv}
    Let $G$ be a connected graph and let $S$ be a set of vertices of $G$. The following statements are equivalent. %
    \begin{enumerate}
        \item There is a GS \cf-tree where all elements of $S$ are leaves.
        \item There is a spanning tree of $G$ where all elements of $S$ are leaves.
        \item The set $V(G) - S$ forms a connected dominating set of $G$.
    \end{enumerate}
\end{theorem}

\begin{proof}
    As every GS \cf-tree is a spanning tree, the first direction holds trivially. 
    
    Let $T$ be a spanning tree of $G$ where the vertices of $S$ are leaves. This implies that $T - S$ forms a subtree of $T$ and, thus, $G - S$ is connected. As the parent of every vertex of $S$ in $T$ is in $G - S$, it holds that $G - S$ forms a connected dominating set of $G$.
    
    Let $V(G) - S$ be a connected dominating set. Let $\sigma$ be a GS ordering of $G - S$ and let $\sigma'$ be a GS ordering of $G$ starting with $\sigma$. Every vertex in $S$ has a neighbor in $G - S$ and, thus, in the $\cf$-tree of $\sigma'$ all vertices of $S$ are leaves.
\end{proof}

Note that the equivalence of the \MaxLeaf{} problem and the connected dominating set problem has been already observed earlier (see, e.g., \cite{fujie2003exact}). \cref{thm:tree-dom-equiv} implies the following.

\begin{corollary}\label{corol:f-max-gs}
    Let $G$ be a connected graph with $n$ vertices. The following statements are equivalent.%
    \begin{enumerate}
        \item There is a GS \cf-tree of $G$ with $\geq k$ leaves.
        \item There is a spanning tree of $G$ with $\geq k$ leaves.
        \item There is a connected dominating set of $G$ with $\leq n - k$ vertices.
    \end{enumerate}
\end{corollary}

As already mentioned in the introduction, there is a wide range of \FPT{} results for \MaxLeaf{} when parameterized by the number of leaves~\cite{bodlaender1993linear,
bonsma2008spanning,fellows1992well,fellows2000coordinatized,kneis2011new,raible2010amortized}. The currently best known algorithm runs in $\Oc^*(3.188^k)$ time~\cite{zehavi2018k-leaf}. By \cref{corol:f-max-gs}, this implies the following.

\begin{corollary}
    \MaxLeafF{} of GS is in \FPT.
\end{corollary}

To complement this result, we will prove next that \MinLeafF{} of GS is also in \FPT{}. We start with the observation that graphs having GS \cf-trees with a small number of leaves have also small pathwidth. The proof follows the same lines as the proofs of Theorem~2.2.1 in~\cite{aazami2008hardness} and Lemma~1 in~\cite{bhyravarapu2025parameterized}.

\begin{lemma}\label{lemma:pw}
    If a graph $G$ has an GS \cf-tree with at most $k$ leaves, then the pathwidth of $G$ is at most $k$. This bound is tight.
\end{lemma}

\begin{proof}
    First observe that the tightness follows from complete graphs. The bound proof follows the same lines as the proofs of Theorem~2.2.1 in \cite{aazami2008hardness} and Lemma~1 in~\cite{bhyravarapu2025parameterized}. Let $\sigma$ be a GS ordering whose \cf-tree $T$ has $k$ leaves. W.l.o.g.~we may assume that $\sigma$ ends with the $k$ leaves. Let $S$ be the set of leaves of $T$. Let $t = |V(G)| - k$. In the following, we construct a path decomposition $(X_0, \dots, X_{2t})$.

    We start with bag $X_0$ which contains the vertices of $S$. Obviously, $|X_0| \leq k$. Now assume we have already compute bag $X_{2i}$ and this bag has size at most $k$. Let $v$ be the vertex at position $t - i$ in $\sigma$ and let $L$ be the set of children of $v$ in $T$. We define $X_{2i+1} = X_{2i} \cup \{v\}$ and $X_{2i+2} = X_{2i+1} \setminus L$. Since $|X_{2i}| \leq k$ and $|L| \geq 1$, it holds that $|X_{2i+1}| \leq k+1$ and $|X_{2i+2}| \leq k$. Repeating this process constructs an ordering of bags $(X_0, \dots, X_{2t})$.

    We claim that this ordering forms a path decomposition of $G$ of width at most $k$. The claim about the width follows directly from the fact that every bag has size at most $k+1$. It is obvious that every vertex is an element of some bag. Furthermore, no vertex is introduced twice, so the bags containing a certain vertex form a contiguous subsequence. It remains to show that every edge $xy$ appears in some bag. Let $x \prec_\sigma y$. Since $\sigma$ ends with the leaves, vertex $y$ is not introduced after $x$ to its first bag. Vertex $y$ is forgotten not before its parent in $T$ is introduced. However, the parent of $y$ is either $x$ or some vertex to the left of $x$. Hence, $x$ is introduced before $y$ is forgotten and, thus, $x$ and $y$ have a common bag.
\end{proof}

Note that -- in contrast to layered searches -- we cannot bound the bandwidth of such graphs.
Consider the path with $n$~vertices and an additional universal vertex. This graph has a GS \cf-tree with two leaves, but its bandwidth is~$\geq \lceil \frac{n}{2} \rceil$.

Due to \cref{lemma:pw}, it is sufficient to show that \MinLeafF{} of GS can be solved in \FPT{} time when parameterized by pathwidth to also show that \MinLeafF{} of GS can be solved in \FPT{} time when parameterized by the number of leaves. Even more general, we will solve \MinLeafF{} of GS parameterized by the treewidth. To this end, we relate the \cf-trees of GS to two concepts that are well studied in the literature.

The first concept are \emph{dominating sequences} which are sequences $(v_1, \dots, v_k)$ of vertices where every vertex dominates some vertex that was not dominated by any vertex to the left of it. More formal, for all $i \in [k]$ it holds that $N\langle v_i\rangle_1 \setminus \bigcup_{j = 1}^{i-1} N\langle v_j\rangle_2 \neq \emptyset$, where $N\langle \cdot \rangle_1$ and $N\langle \cdot \rangle_2$ are replaced by $N(\cdot)$ or $N[\cdot]$, depending on the type of dominating sequence (see~\cite{bresar2017grundy,bresar2014dominating,bresar2016total}). Here, we are interested in so-called \emph{$Z$-sequences}, where $N\langle \cdot \rangle_1 = N(\cdot)$ and $N\langle \cdot \rangle_2 = N[\cdot]$. These sequences were introduced in \cite{bresar2017grundy}. We adapt them in such a way that the sequence forms a prefix of a GS ordering.

\begin{definition}
Let $G$ be a graph. A sequence $(v_1, \dots, v_k)$ of vertices is a \emph{generic $Z$-sequence} of $G$ if it is the prefix of a GS ordering of $G$ and if for all $i \in [k]$ it holds that $N(v_i) \setminus \bigcup_{j=1}^{i-1} N[v_j]$ is not empty.
\end{definition}

A concept strongly related to dominating sequences are \emph{zero forcing sets}. Such a set is a subset of vertices which are initially colored blue while all other vertices are colored white. Further, there is a set of certain color change rules that allow to change the color of white vertices to blue. There are many different color change rule schemes in the literature~\cite{barioli2013parameters,hogben2022inverse,lin2019zero}. One of these rules, called \emph{$Z$-rule} in~\cite{scheffler2025parameterized}, allows a blue vertex $v$ to color a white neighbor $w$ blue if and only if $w$ is the only white neighbor of $v$. We adapt this rule as follows.

\begin{description}
    \item[Z$^*$-rule] If $v$ is a blue vertex and $v$ has exactly one white neighbor $w$ and $w$ is either the unique white vertex in $G$ or $w$ has at least one white neighbor $x$, then change the color of $w$ to blue. We write $v \z w$.
\end{description}

Given a graph $G$, we define a set $S \subseteq V(G)$ to be an \emph{$Z^*$-forcing set} of $G$ if the following procedure is able to color all vertices of $G$ blue:
\begin{enumerate}
    \item Color the vertices of $S$ blue and all vertices of $V(G) \setminus S$ white.
    \item Iteratively apply the $Z^*$-rule to the vertices of $G$.
\end{enumerate}

The concepts generic $Z$-sequences, $Z^*$-forcing sets and minimum leaf \cf-trees of GS are strongly related. 

\begin{theorem}\label{thm:forcing-domination-tree}
    The following statements are equivalent for an $n$-vertex graph~$G$ with $n \geq 2$:
    \begin{enumerate}
        \item There is a GS \cf-tree of $G$ with $\leq k$ leaves.\label{item:max-gs1}
        \item There is a $Z^*$-forcing set of size $\leq k$.\label{item:max-gs2}
        \item There is a generic $Z$-sequence of length $\geq n - k$.\label{item:max-gs3}
    \end{enumerate}
\end{theorem}

\begin{proof}
    First we prove that \ref{item:max-gs3} implies \ref{item:max-gs1}. Let $(v_1, \dots, v_{n-k})$ be a generic $Z$-sequence of $G$ of length $n-k$. We extend this sequence to a GS ordering $\sigma = (v_1, \dots, v_n)$. Then in the \cf-tree $T$ of $\sigma$ every vertex $v_i \in \{v_1, \dots, v_{n-k}\}$ is an internal vertex since all vertices in $N(v_i) \setminus \bigcup_{j=1}^{i-1} N[v_j]$ are children of $v_i$ in $T$. Thus, $T$ has at most $k$ leaves.

    Next, we show that \ref{item:max-gs1} implies \ref{item:max-gs2}. Let $\sigma = (v_1, \dots, v_n)$ be a GS ordering whose \cf-tree $T$ has at most $k$ leaves. W.l.o.g. we may assume that $\sigma$ ends with the leaves. Let $L$ be the set of leaves of $T$. Note that $|L| \leq k$. We claim that $L$ is a $Z^*$-forcing set of $G$. We color the $v_i$ blue in descending order of their indices starting with $i = n - |L|$. Since $v_i$ is not a leaf in $T$, it has a child $v_j$ in $T$. Since $j > i$, vertex $v_j$ is already colored blue and no vertex to the left of $v_i$ is adjacent to $v_j$. Thus, $v_i$ is the only white neighbor of $v_j$ and, if $i \neq 1$, then $v_i$ has a (white) parent $v_\ell$ in $T$ with $\ell < i$. Hence, we can apply rule $v_j \z v_i$.

    Finally, we show that \ref{item:max-gs2} implies \ref{item:max-gs3}. Let $S$ be a $Z^*$-forcing set of size $\leq k$ and let $(R_1, \dots, R_\ell)$ be the sequence of the applied rules with $R_i = u_i \z v_i$. Note that $|S| + \ell = n$ and, thus, $\ell \geq n - k$. We claim that $\sigma = (v_\ell, \dots, v_1)$ is a generic $Z$-sequence of $G$. Every vertex $v_i \neq v_\ell$ had a white neighbor $v_j$ when rule $R_i$ was applied. Vertex $v_j$ was colored blue after $v_i$ and, thus, $j > i$ and $v_j \prec_\sigma v_i$. This implies that $v_i$ has a neighbor to its left in $\sigma$ and, therefore, $\sigma$ is the prefix of a GS ordering of $G$. Since all elements of $N[u_i] \setminus \{v_i\}$ were blue when rule $R_i$ was applied, all these vertices do not lie to the left of $v_i$ in $\sigma$. Hence, $u_i \in N(v_i) \setminus \bigcup_{v_j \prec_\sigma v_i} N[v_j]$ and $\sigma$ is a generic $Z$-sequence of length $\geq n - k$.
\end{proof}

The problem \ZFS{} which uses the $Z$-rule mentioned above can be solved in \FPT{} time when parameterized by the treewidth of $G$~\cite{bhyravarapu2025parameterized,scheffler2025parameterized}. We adapt the algorithm given in~\cite{scheffler2025parameterized} to also find a minimal $Z^*$-forcing set in \FPT{} time when parameterized by the treewidth of $G$. As this algorithm and the proof of correctness is mainly a rephrasing of the results given in~\cite{scheffler2025parameterized}, we deferred them to the appendix (see \cref{appendix}).

\begin{theorem}\label{thm:tw}
    Given an $n$-vertex graph $G$ of treewidth $d$, we can compute a smallest $Z^*$-forcing set of $G$ in $2^{\Oc(d^2)} \cdot n$ time.
\end{theorem}

As pointed out above, this implies an \FPT{} algorithm for \MinLeafF{}.

\begin{theorem}\label{thm:min-leaf-fpt}
    \MinLeafF{} of GS can be solved in $2^{\Oc(d^2)} \cdot n$ time where $d$ is the treewidth of the graph; it can be solved in $2^{\Oc(k^2)} \cdot n$ time where $k$ is the number of leaves of the tree.
\end{theorem}

\begin{proof}
     The first statement follows directly from \cref{thm:forcing-domination-tree,thm:tw}. 
     
     If we want to solve the problem parameterized by $k$, then we apply Korhonen's 2-approximation for treewidth~\cite{korhonen2021single} to the input $(G,k)$ which needs $2^{\Oc(k)} \cdot n$ time. If it outputs that the treewidth of $G$ is larger than~$k$, then there is no GS \cf-tree with $\leq k$ leaves, due to \cref{lemma:pw}. Otherwise, we can solve the problem in $2^{\Oc(k^2)} \cdot n$ time using our algorithm for the treewidth case.
\end{proof}

\section{Number of Internal Vertices as Parameter}\label{sec:internals}

\subsection{Hardness for Search Trees with Few Internal Vertices}

As we have seen in \cref{corol:f-max-gs}, \MinInF{} is equivalent to \ConDom{}. This problem is known to be \NP-complete and \WT-complete when parameterized by the size of the dominating set~\cite{cygan2015param}. This implies that \MinInF{} of GS is \WT-complete and \NP-complete and \MaxLeafF{} is \NP-complete. As shown next, we can generalize these hardness results to all connected graph searches that are clique starters.

\begin{theorem}\label{thm:max-leaf-F-NPc}
    Let $\A$ be a connected graph search that is a clique starter. Then 
        \begin{enumerate}
        \item \MaxLeafF{} of $\A$ is \NP-hard on split graphs and
        \item \MinInF{} of $\A$ is \WT-hard and \NP-hard on split graphs.
    \end{enumerate}
\end{theorem}

\begin{proof}
    We reduce from \textsc{Set Cover}. An instance of this problem consists of a finite set $U = \{u_1, \dots, u_n\}$ and a family $\mathcal{W} = \{W_1, \dots, W_p\}$ of subsets of $U$ as well as an integer $k \leq p$. The question is whether there are $\ell \leq k$ sets $W_{i_1}, \dots, W_{i_\ell}$ such that $\bigcup_{j=1}^\ell W_{i_j} = U$. This problem is both \NP-complete~\cite{karp21} and \WT-complete~\cite{downey1995fixed} when parameterized by $k$. W.l.o.g.~we may assume that $\bigcup_{j=1}^{p} W_{j} = U$, i.e., every element of $U$ is in at least one set $W_j$. Furthermore, we assume that no element of $U$ is contained in every set $W_i \in \mathcal{W}$ since such an element has no influence on the size of the minimum set cover.

    We build the following split graph $G$. For every $W_i$, we construct a vertex $c_i$ and for every $u_i$ we construct a vertex $a_i$. We call the set of $c$-vertices $C$ and the set of $a$-vertices $A$. We add edges to $G$ such that the set $C$ induces a clique while the set $A$ induces an independent set. Furthermore, vertex $c_i$ is made adjacent to vertex $a_j$ if and only if the set $W_i$ contains the element $u_j$. We claim that there is a set cover of size less than or equal to $k$ if and only if there is an \cf-tree of search $\A$ on $G$ with no more than $k$ internal vertices. 

    First assume there is a set cover $W_{i_1}, \dots W_{i_\ell}$ with $\ell \leq k$. Since $\A$ is a clique starter, there is an $\A$-ordering $\sigma$ with the prefix $(c_{i_1}, \dots, c_{i_\ell})$. Since we consider a set cover, all vertices that are not part of this prefix have a neighbor in the prefix. This implies that all these vertices are leaves in the $\cf$-tree of $\sigma$. Hence, this \cf-tree has at most $k$ internal vertices.

    Now assume that there is an $\A$-ordering $\sigma$ whose $\cf$-tree $T$ has at most $k$ internal vertices. Let $S$ be the set of internal vertices of $T$ that are part of $C$. We claim that the family $\mathcal{W}' = \{W_i \mid c_i \in S\}$ forms a set cover of $U$. Let $u_j$ be an arbitrary element of $U$. First assume that the vertex $a_j$ is not the root of $T$. Then the parent of $a_j$ is in $C$. Let us call the parent $c_i$. As $c_i$ is an internal vertex of $T$, it is contained in $S$. Thus, $W_i \in \mathcal{W}'$ and $u_j$ is covered by~$\mathcal{W}'$. Now assume that $a_j$ is the root of $T$. As $\A$ is a connected search, the second vertex of $\sigma$ is a neighbor of $a_j$ and, thus, part of $C$. Due to our assumption, $u_j$ is not contained in every set $W_i$ and, thus, there is at least one vertex $c_\ell \in C$ that is not adjacent to $a_j$. Let $c_s$ be the second vertex of $\sigma$. Vertex $c_s$ is different from $c_\ell$ and has $c_\ell$ as its child in the \cf-tree of $\sigma$. Hence, $c_s$ is contained in $S$ and $u_j$ is covered by $W_s$.
\end{proof}

Note that we leave open whether there are clique starters other than GS for which \MinInF{} is not just \WT-hard but also \WT-complete. The main difference between many of these searches and GS is the fact that the previously visited vertices have a significant influence on the ordering of the following vertices. In particular, the ordering of leaves of the \cf-tree might influence the ordering of some internal vertices (see end of \cref{sec:inner-algorithm}). Nevertheless, the problem is \NP-complete for a search as long as its orderings can be recognized in polynomial time. This holds, e.g., for BFS. Adapting proofs given in \cite{scheffler2025leaves}, we can strengthen this result for LBFS.

\begin{theorem}\label{thm:np-lbfs}
    For every fixed value $k\geq3$, \MinInF{} of LBFS is \NP-complete on weakly chordal graphs.
\end{theorem}

\begin{proof}
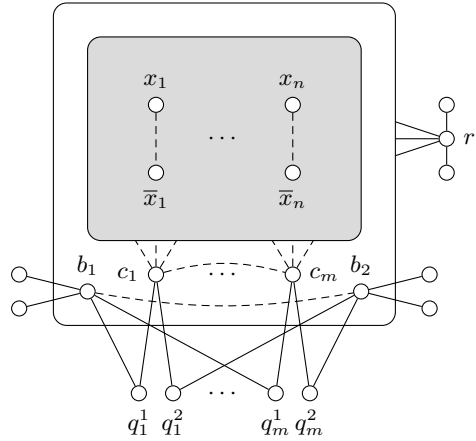
\begin{figure}
	\centering
    	\begin{tikzpicture}[scale=0.90, vertex/.style={inner sep=2pt,draw,circle, fill=white},
 	noedge/.style={dash pattern=on 3pt off 2pt}
 	]
 	\small
 	\draw[rounded corners=5pt, fill=white] (-1.5,1.5) rectangle (3.5,-3.25);
 	\draw[rounded corners=5pt, fill=black!14!white] (-1,1) rectangle (3,-2);

 	\node[vertex, label={90:$x_1$}] (x1) at (0,0) {};
 	\node[vertex, label={-90:$\overline{x}_1$}] (nx1) at (0,-1) {};
 	\node[vertex, label={90:$x_n$}] (xn) at (2,0) {};
 	\node[vertex, label={-90:$\overline{x}_n$}] (nxn) at (2,-1) {};
 	
 	\node[vertex, label={0:$r$}] (r) at (4.25,-0.5) {}
 	edge (3.5,-0.25)
    edge (3.5,-0.5)
    edge (3.5,-0.75);
 	
 	\node[vertex, label={180:$c_1$}] (c1) at (0,-2.5) {}
 	edge[noedge] (-0.3,-2.)
    edge[noedge] (0,-2.)
    edge[noedge] (0.3,-2.);
 	\node[vertex, label={0:$c_m$}] (cm) at (2,-2.5) {}
 	edge[noedge] (1.7,-2.)
    edge[noedge] (2,-2.)
    edge[noedge] (2.3,-2.);
    
    \draw[noedge, bend angle=15, bend left] (c1) to (cm);
    
    \node[vertex, label={-90:$q^1_1$}] (q11) at (-0.25,-4.25) {}
    edge (c1);
    \node[vertex, label={-90:$q^2_1$}] (q12) at (0.25,-4.25) {}
    edge (c1);
    \node[vertex, label={-90:$q^1_m$}] (qm1) at (1.75,-4.25) {}
    edge (cm);
    \node[vertex, label={-90:$q^2_m$}] (qm2) at (2.25,-4.25) {}
    edge (cm);
 	
 	\draw[noedge] (x1) -- (nx1);
 	\draw[noedge] (xn) -- (nxn);

    \node[vertex] (lr1) at (4.25,-0) {}
    edge (r);

    \node[vertex] (lr2) at (4.25,-1) {}
    edge (r);
    
    \node[vertex] (lb11) at (-2,-2.5) {};
    \node[vertex] (lb12) at (-2,-3) {};

    \node[vertex] (lb21) at (4,-2.5) {};
    \node[vertex] (lb22) at (4,-3) {};
 	
 	\node[vertex, label={90:$b_1$}] (b1) at (-1,-2.75) {}
    edge (q11)
    edge (qm1)
    edge (lb11)
    edge (lb12);

    \node[vertex, label={90:$b_2$}] (b2) at (3,-2.75) {}
    edge (q12)
    edge (qm2)
    edge (lb21)
    edge (lb22);

    \draw[noedge, bend angle=10, bend right] (b1) to (b2);
 	
 	\node (p) at (1,-0.5) {$\cdots$};
 	\node (p) at (1,-2.5) {$\cdots$};
 	\node (p) at (1,-4.25) {$\cdots$};

 	\end{tikzpicture}
 	\caption{The construction of the proof of \cref{thm:np-lbfs}. In the boxes, only non-edges are displayed by dashed lines. The clause vertices have exactly three non-neighbors in the gray box. The connection of a vertex with a box means that the vertex is adjacent to all vertices in the box.}\label{fig:np-lbfs}
\end{figure}
We adapt the proofs of Theorems 4.16 and 4.17 in {\cite{scheffler2025leaves}}. 
We reduce from 3-SAT. Let $\I$ be an instance of 3-SAT. We construct the corresponding graph $G(\I)$ as follows (see~Figure~\ref{fig:np-lbfs}). Let $ X=\{x_1, \dots, x_k,\overline{x}_1,\ldots,\overline{x}_k\} $ be the set of vertices representing the literals of $ \I $. The graph $G(\I)[X]$ forms the complement of the matching in which $x_i$ is matched to $ \overline{x}_i $ for every $ i \in \{1,\ldots, k\} $. Let $ C = \{c_1, \ldots ,c_\ell\} $ be the set of vertices representing the clauses of $ \I $. The set $ C $ forms an independent set in $G(\I)$ and every clause vertex $ c_i $ is adjacent to each vertex of $ X $ whose corresponding literal is \emph{not} contained in the clause associated with $ c_i $. We have two vertices $b_1$ and $b_2$ that are adjacent to all vertices in $X \cup C$ but are not adjacent to each other. Additionally, we add a vertex $r$ that is adjacent to all vertices in $X \cup C \cup \{b_1,b_2\}$. For every vertex $c_i$, we add two vertices $q^1_i$ and $q^2_i$. For $j \in \{1,2\}$, vertex $q^j_i$ is adjacent to $c_i$ as well as $b_j$. Finally, we append to each of the vertices $r$, $b_1$ and $b_2$ two leaves.

We claim that $\I$ has a satisfying assignment if and only if $G(\I)$ has a LBFS ordering whose \cf-tree has exactly three internal vertices. First assume that $\I$ has a satisfying assignment $\B$. We start the LBFS ordering in $r$. Then we visit all literal vertices whose literal is satisfied by $\B$. Note that this is possible since the visited vertices form a clique in $G[\I]$. As $\B$ is satisfying, every vertex $c_i$ has some visited non-neighbor. The same holds for all remaining literal vertices since their inverse literal vertex, which is not adjacent to them, has been visited. This implies that $b_1$ and $b_2$ are the only vertices that can be visited next. Visiting these two vertices fixes the parents of all vertices. In fact, every vertex of $X \cup C \cup \{b_1,b_2\}$ has parent $r$, every vertex $q^1_i$ has parent $b_1$ and every vertex $q^2_i$ has parent $b_2$. The leaves have their respective neighbor as parent. Therefore, every LBFS ordering following the described prefix has an \cf-tree with exactly three internal vertices.

Now assume $\sigma$ is an LBFS ordering of $G(\I)$ whose \cf-tree has exactly three internal vertices. Since $r$, $b_1$ and $b_2$ have two adjacent leaves, they must be the unique internal vertices. In particular, this implies that both $b_1$ and $b_2$ are to the left of all vertices $c_i$ since otherwise $c_i$ is the parent of either $q^1_i$ or $q^2_i$. W.l.o.g.~we may assume that $b_1 \prec_\sigma b_2$. Therefore, either $r$ or $b_1$ is the start vertex of $\sigma$. If this is $b_1$, then $b_2$ is visited after all clause vertices; a contradiction. Thus, we know that $\sigma$ starts in $r$. Since $b_2$ is not adjacent to $b_1$ but all clause vertices are adjacent to $b_1$, all vertices $c_i$ must have some non-neighbor to the left $b_1$ and this non-neighbor cannot be another clause vertex as these are all to the right of $b_1$. Since this non-neighbor must be adjacent to $r$, the only possible vertices are the literal vertices whose literals appear in the clause of $c_i$. Thus, at least one of these three literal vertices is to the left of $b_1$ in $\sigma$. Note that not both literal vertices of the same variable can be to the left of $b_1$ since there are non-adjacent and, thus, one of the two will be visited after $b_1$. Summarizing, the literal vertices to the left of $b_1$ induce an assignment of a subset of the variables of $\I$ that satisfied~$\I$. This concludes the proof of correctness of the reduction.

To see that the graph $G[\I]$ is weakly chordal, we apply \cref{lemma:simplicial} to delete vertices whose existence or non-existence does not influence the property of being weakly chordal. The leaves are simplicial and, thus, can be deleted. The same holds for all the vertices $q^1_i$ and $q^2_i$. Now $r$ is adjacent to all vertices and can also be deleted. Vertices $b_1$ and $b_2$ are adjacent to all vertices but one and can also be deleted. Now we observe that every pair $\{x_i,\overline{x}_i\}$ forms a two-pair of the remaining graph. Hence, we can add the edge between them, due to \cref{lemma:2-pair}. The resulting graph is a split graph and, thus, weakly chordal. This implies that $G(\I)$ is also weakly chordal, due to \cref{lemma:simplicial,lemma:2-pair}.

For values of $k$ that are larger than three, we can use the same construction and just append a path containing $k - 2$ new vertices to $r$.
\end{proof}

Note that the bound for $k$ is tight. For all clique starters, it is easy to see that they have an \cf-tree with at most $k \leq 2$ internal vertices if and only if the graph has a connected dominating set of size at most~$k$.

\subsection{Hardness for Search Trees with Many Internal Vertices.}

To prove that \MaxInF{} is \W-hard for all clique starters, we consider special vertex orderings of graphs that are related to the generic $Z$-sequences used in~\cref{sec:leaves}. A~sequence $\sigma = (v_1, \dots, v_k)$ of vertices of a graph $G$ is called \emph{total dominating sequence}\footnote{Note that a total dominating sequence does not necessarily dominate the graph. Nevertheless, it can always be extended to such a sequence.} of $G$ if for every $i \in [k]$ it holds that $N(v_i) \setminus \bigcup_{j=1}^{i-1} N(v_j)$ is not empty. This notion was introduced by Bre\v{s}ar et al.~\cite{bresar2016total} where it is shown that the corresponding optimization problem that looks for a total dominating sequence of size at least $k$ is $\NP$-complete. Scheffler~\cite{scheffler2025parameterized} considered a bipartite version called \BIP{}. Here, a bipartite graph $G$ is given with $V(G) = X \dot \cup Y$ and one asks whether there is a total dominating sequence of size $k$ that only contains vertices of $X$. It is shown in~\cite{scheffler2025parameterized} that \BIP{} is both \W-complete and \NP-complete. This allows us to prove the following result.

\begin{theorem}\label{thm:max-int-f-w-hard}
      Let $\A$ be a connected graph search that is a clique starter. Then 
        \begin{enumerate}
        \item \MaxInF{} of $\A$ is \W-hard and \NP-hard on split graphs and
        \item \MinLeafF{} of $\A$ is \NP-hard on split graphs.
    \end{enumerate}
\end{theorem}

\begin{proof}
    We reduce from \BIP. Let $G$ be a bipartite graph with $V(G) =X \dot\cup Y$ and $k \in \N$. W.l.o.g.~we may assume that $X$ does not contain isolated vertices as such vertices can never be part of a total dominating sequence. We construct the graph $G'$ from $G$ as follows: We add a vertex $r$ to $X$ and make $X$ to a clique. Note that $G'$ is a split graph. We claim that $G$ has a total dominating sequence containing $k$ vertices of $X$ if and only if there is an $\A$-ordering of $G'$ whose $\cf$-tree has $\geq k + 1$ internal vertices.

    First assume that $(v_1, \dots, v_k)$ is a total dominating sequence of $G$ containing $k$ vertices of $X$. Consider an $\A$-ordering $\sigma$ starting with the prefix $(r,v_1, \dots, v_k)$. Since $\A$ is a clique starter, there exists such an $\A$-ordering. Let $T$ be the \cf-tree of $\sigma$. Obviously, $r$ is an internal vertex of $T$. As $(v_1, \dots, v_k)$ is a total dominating sequence, for each vertex $v_i$ with $i \in [k]$, there is a vertex in the set $N(v_i) \setminus (N(r) \cup \bigcup_{j=1}^{i-1} N(v_i))$. This vertex is a child of $v_i$ in $T$ and, thus, $v_i$ is an internal vertex of $T$.

    Now assume that there is an $\A$-ordering $\sigma$ of $G'$ whose $\cf$-tree has at least $k+1$ internal vertices. If the start vertex is an element of $X$, then all vertices of $X$ are its children. Thus, no vertex of $Y$ can be an internal vertex of $T$. If the start vertex is an element of $Y \cup \{r\}$, then the second vertex must be an element of $X \setminus \{r\}$ since $\A$ is a connected graph search. Again, no other vertex of $Y$ can be an internal vertex of $T$ as all vertices of $X$ are children of either the start vertex or the second vertex. Summarizing, there are at least $k$ internal vertices in $X \setminus \{r\}$. Let $(v_1, \dots, v_k)$ be the ordering of the first $k$ internal vertices of $X \setminus \{r\}$ in $\sigma$. We claim that $(v_1, \dots, v_k)$ is a total dominating sequence of $G$. First consider $v_1$. Since $v_1$ is not isolated in $G$, it has at least one neighbor in $Y$ and, thus, $N_G(v_1) \setminus \bigcup_{j=1}^0 N_G(v_j)$ is not empty. So consider a vertex $v_i$ with $i \geq 2$. This vertex has a child $x$ in $T$. Vertex $x$ is not adjacent to any vertex $v_j$ with $j < i$. Thus, $x$ cannot be an element of $X \cup \{r\}$ since these vertices are adjacent to $v_1$. So $x \in Y$ and, thus, $x \in N_G(v_i) \setminus \bigcup_{j=1}^{i-1} N_G(v_j)$.
\end{proof}

Similar as for the minimization problem, we can show that \MaxInF{} is \W-complete at least for GS.

\begin{theorem}\label{thm:gs-w1-complete}
    \MaxInF{} of GS is \W-complete.
\end{theorem}

\begin{proof}
    We have to give a \FPT{} reduction from \MaxInF{} of GS to \WCS{} for circuits of weft~1 and bounded depth. We use a similar idea as has been used in \cite{scheffler2025parameterized} to show \W-completeness of Grundy domination problems. We use three types of variables:

    \begin{itemize}
        \item $x(v,i)$ is true if and only if vertex $v$ is the $i$-th vertex in the GS ordering that is an internal vertex of the $\cf$-tree.
        \item $y(v,i)$ is true if $v$ is a child of the $i$-th vertex in the GS ordering that is an internal vertex of the $\cf$-tree.
        \item $z(j,i)$ is true if the $i$-th internal vertex in the GS ordering is adjacent to the $j$-th internal vertex with $j < i$.
    \end{itemize}

    First, we want to ensure that for every $i$ there is at most one variable of every type that is set true. To this end, we use the following three formulas:

    \begin{align*}
        X &:= \bigwedge_{i \in [k]} \bigwedge_{\substack{v,w \in V(G) \\ v \neq w}} \big(\lnot x(v,i) \lor \lnot x(w,i)\big) \\
        Y &:= \bigwedge_{i \in [k]} \bigwedge_{\substack{v,w \in V(G) \\ v \neq w}} \big(\lnot y(v,i) \lor \lnot y(w,i)\big) \\
        Z &:= \bigwedge_{\substack{i \in [k] \setminus \{1\}}} \bigwedge_{\substack{j,j' \in [i-1] \\ j \neq j'}} \big(\lnot z(j,i) \lor \lnot z(j',i)\big)
    \end{align*}

    Next, we ensure that no vertex appears twice in the ordering of the internal vertices.
    \begin{align*}
        X' &:= \bigwedge_{v \in V(G)} \bigwedge_{\substack{i,j \in [k] \\ i \neq j}} \big(\lnot x(v,i) \lor \lnot x(v,j)\big)
    \end{align*}

    Now, we have to ensure that the children are in fact children in the $\cf$-tree.
    \begin{align*}
        A &:= \bigwedge_{i \in [k]} \bigwedge_{j \in [i-1]} \bigwedge_{v \in V(G)} \bigwedge_{w \in N[v]} \big( \lnot y(w,i) \lor \lnot x(v,j) \big) \\
        B &:= \bigwedge_{i \in [k]} \bigwedge_{v \in V(G)} \bigwedge_{w \notin N(v)} \big( \lnot y(w,i) \lor \lnot x(v,i) \big)
    \end{align*}

    Formula $A$ ensures that if $w$ is the child of the $i$-th internal vertex, then it is neither adjacent nor equal to one internal vertex before the $i$-th vertex. Formula $B$ ensures that if $w$ is the child of the $i$-th internal vertex, then the $i$-th internal vertex is neither $w$ nor a non-neighbor of $w$.

    Finally, we have to ensure that the assignment of the $z$-variables match the adjacencies in the graph.
    \begin{align*}
        C := \bigwedge_{i \in [k] \setminus \{1\}} \bigwedge_{j \in [i-1]} \bigwedge_{v \in V(G)} \bigwedge_{w \notin N(v)} \big(\lnot z(j,i) \lor \lnot x(v,j) \lor \lnot x(w,i)\big)
    \end{align*}

    The final formula is $F := X \land Y \land Z \land A \land B \land C$. It is easy to check that the circuit has weft~1 and depth $\leq 4$.

    We claim that the circuit has a satisfying assignment of weight $3k - 1$ if and only if the graph $G$ has a $\cf$-tree of GS with at least $k$ internal vertices. 
    
    First assume there is a satisfying assignment of weight $3k - 1$. As mentioned, $X$, $Y$, and $Z$ ensure that for every $i$, there is at most one true $x$-, $y$-, and $z$-variable (while for $i = 1$ there are no $z$-variables). So if the assignment has weight $3k - 1$, there are exactly $k$ true $x-$variables and $k$ true $y$-variables as well as $k-1$ true $z$-variables. Due to $X'$, the $k$ vertices with a true $x$-variable are distinct. Let $v_1, \dots, v_k$ be the vertices such that $x(v_i, i)$ is true for all $i \in [k]$. For every $i \in [k] \setminus \{1\}$, there is a $j < i$ such that $z(j,i)$ is true. Due to $C$, it then holds that $v_i$ and $v_j$ are adjacent. Thus, $\sigma = (v_1, \dots, v_k)$ is the prefix of a GS ordering. It remains to show that the \cf-tree $T$ of every GS ordering with prefix $\sigma$ has at least $k$ internal vertices. For every $i \in [k]$, there is a $w$ such that $y(w,i)$ is true. Due to $A$, none of the vertices $v_1, \dots, v_{i-1}$ is adjacent or equal to $w$. Due to $B$, $v_i$ is adjacent to $w$. Thus, $w$ is a child of $v_i$ in $T$ and $v_i$ is an internal vertex of $T$. Hence, $T$ has at least $k$ internal vertices.

    Now assume that $\sigma$ is a GS ordering of $G$ whose $\cf$-tree $T$ has at least $k$ internal vertices. Let $(v_1, \dots, v_k)$ be the ordering of the first internal vertices of $T$ in $\sigma$. We set $x(v_i,i)$ to true for every $i \in [k]$. Let for every $i \in [k]$, $w_i$ be a child of $v_i$ in $T$. We set $y(w_i,i)$ to true. Note that the parent of vertex $v_i$ is an internal vertex and is to the left of $v_i$ in $\sigma$. Hence, the parent of $v_i$ with $i > 2$ is some vertex $v_j$ with $j < i$. We set $z(j,i)$ to true. All other variables are set to false. It is obvious that this assignment has weight $3k-1$. Furthermore, it is easy to check that it satisfies $F$.
\end{proof}

Again, we have to leave open whether \MaxInF{} of other clique starters is in \W. Nevertheless, \NP-completeness holds as long as the orderings of the search can be recognized in polynomial time.

\subsection{Algorithms}\label{sec:inner-algorithm}

The results of the last section imply that we do not have to look for \FPT{} algorithms for \MinInF{} or \MaxInF{} for any graph search that is a clique starter. Here, we will consider the existence of \XP{} algorithms for these problems. As seen above (see \cref{thm:tree-dom-equiv}), \MinInF{} of GS is equivalent to \ConDom{} which can be solved straightforwardly in $n^{\Oc(k)}$ time. The proof of \cref{thm:gs-w1-complete} presents a reduction of \MaxInF{} of GS to \WCS{} that increases the parameter only linearly. Since \WCS{} can be solved in $n^{\Oc(k)}$ time, this implies the following.

\begin{theorem}
    \MinInF{} and \MaxInF{} of GS can be solved in $n^{\Oc(k)}$ time.
\end{theorem}

Next we will give a similar result for BFS. To this end, we show that a BFS \cf-tree with a certain set of internal vertices $S$ can be computed using only the knowledge about the ordering of the vertices in $S$.

\begin{lemma}\label{lemma:bfs-xp}
    Let $T_\sigma$ be the \cf-tree of some BFS ordering $\sigma$ and let $S$ be the set of the first $k$ internal vertices of $T_\sigma$. Let $\rho'$ be the ordering of $S$ in $\sigma$ and let $\rho$ be a vertex ordering of $G$ starting with $\rho'$. Let $\tau$ be the $\text{BFS}^+_\rho$ ordering of $G$ and let $T_\tau$ be its \cf-tree. Then, every $v \in S$ has the same children in $T_\sigma$ and $T_\tau$.
\end{lemma}

\begin{proof}
     First observe that the start vertex of $\tau$ is the first vertex of $\rho$ and the first vertex of $\rho$ is the first vertex of $\rho'$. As the start vertex is by definition an internal vertex of the \cf-tree, it holds that the start vertex of $\sigma$ is an element of $S$. As $\sigma$ contains $\rho'$ as subordering, it must hold that $\sigma$ also starts with the first vertex of $\rho'$. Thus, $\sigma$ and $\tau$ start with the same vertex and, consequently, the layers of $\sigma$ and $\tau$ are identical.
      
    We show by induction that for every layer the following statements are true:

    \begin{enumerate}[(T1)]
        \item If $v$ and $w$ are in layer~$i$, $v \in S$ and $v \prec_\sigma w$, then $v \prec_\tau w$.\label{item:bfs-xp1}
        \item Every vertex of $S$ in layer~$i$ has the same set of children in $T_\sigma$ as in $T_\tau$.\label{item:bfs-xp2}
    \end{enumerate}

    Note that these statements are trivially true for the zeroth layer as $\sigma$ and $\tau$ have the same start vertex $r$ and all elements of the first layer are the children of $r$. So we may assume that the two statements hold for some layer~$i$. First, we prove statement~(T\ref{item:bfs-xp1}) for layer~$i+1$. Let $v$ and $w$ be in layer~$i+1$ such that $v \in S$ and $v \prec_\sigma w$.
    Since layer~$i+1$ contains at least one of the first $k$ internal vertices of $T$, the parents of $v$ and $w$ are also in $S$. Since (T\ref{item:bfs-xp2}) holds for layer~$i$, $v$ has the same parent $p_v$ in $T_\tau$ and $T_\sigma$ and $w$ has the same parent $p_w$ in $T_\tau$ and $T_\sigma$. If $p_v$ and $p_w$ are not identical, then $p_v$ is before $p_w$ in $\sigma$. As (T\ref{item:bfs-xp1}) holds for layer~$i$, $p_v$ is before $p_w$ also in $\tau$ and, thus, it follows that $v \prec_\tau w$. So we may assume that $p_v$ and $p_w$ are the same vertex. Since $v \prec_\sigma w$ and $v \in S$, it holds that $v \prec_\rho w$ (either $w$ is to the right of $v$ in $\rho'$ or $w$ is not in $S$ and, thus, to the right of $v$ in $\sigma$). Thus, the tie-breaking of $\text{BFS}^+_\rho$ puts $v$ before $w$ in $\tau$.

    It remains to show that (T\ref{item:bfs-xp2}) holds for layer~$i+1$. Let $x$ be an element of layer~$i+2$ and assume, for contradiction, that the parent of $x$ in $T_\sigma$ is $p_{\sigma}$, the parent of $x$ in $T_\tau$ is $p_{\tau} \neq p_{\sigma}$ and at least one of $p_{\sigma}$ and $p_{\tau}$ is in $S$. Observe, that $p_{\sigma}$ has to be an element of $S$, because either $p_{\sigma}$ is the only element of $\{p_{\sigma},p_{\tau}\}$ contained in $S$, or, if $p_{\tau}$ is in $S$, then $p_{\sigma}$ has to be in $S$ as well, because it is an internal vertex of $T_\sigma$, preceding $p_{\tau} \in S$ in $\sigma$. Furthermore, $p_{\sigma} \prec_\sigma p_{\tau}$ but $p_{\tau} \prec_\tau p_{\sigma}$. This contradicts (T\ref{item:bfs-xp1}) for layer~$i+1$.
\end{proof}

Using this result, we can give \XP{} algorithms for \MinInF{} and \MaxInF{} of BFS.

\begin{theorem}\label{thm:bfs-xp}
    \MinInF{} and \MaxInF{} of BFS can be solved in $\Oc(n^{k+2})$ time.
\end{theorem}

\begin{proof}
    For \MinInF{}, we consider every choice of $\leq k$ internal vertices $S$ and every possible ordering $\rho'$ of the vertices in $S$. There are $\Oc(n^k)$ choices of internal vertices and orderings. So in total we have to check $\Oc(n^k)$ choices. For every of these choices, we compute the $\text{BFS}^+_\rho$ ordering $\sigma$ for some ordering $\rho$ starting with $\rho'$. If $\sigma$ has an \cf-tree with $\leq k$ internal vertices, we stop and return yes. If not, then we consider the next choice. If we check all choices without finding a correct \cf-tree, then we reject the input. If there is a BFS \cf-tree with at most $k$ internal vertices $S$, then all vertices outside of $S$ must be children of some vertex in $S$. Then, due to \cref{lemma:bfs-xp}, this also holds in the BFS \cf-tree constructed by our procedure for the respective ordering $\rho'$ of $S$.

    The algorithm for \MaxInF{} works analogously. Computing the BFS$^+_\rho$ ordering and its $\cf$-tree as well as the number of internal vertices can be done in $\Oc(n^2)$ time. This implies the given running time bound.
\end{proof}

Due to \cref{thm:np-lbfs}, we cannot give such an algorithm for \MinInF{} of LBFS, unless $\sP = \NP$. This difference in the complexity is in line with the complexity of the \cf-tree recognition problem which is polynomial-time solvable for GS and BFS~\cite{hagerup1985recognition,scheffler2025partial} but \NP-hard for LBFS~\cite{beisegel2021recognition}.

An adaptation of \cref{thm:bfs-xp} for \MaxInF{} of LBFS seems also to be difficult. 
In contrast to BFS, the ordering of the leaves of the tree has a significant influence on the ordering of the following layers (see~\cref{fig:counterexample-lbfs}).

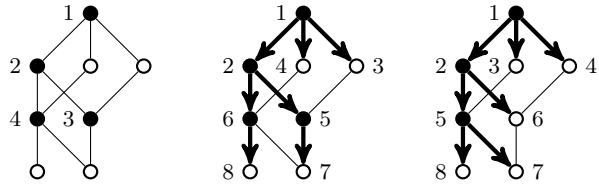
\begin{figure}
    \centering
    \begin{tikzpicture}[xscale=0.7, yscale=0.7]
\footnotesize
\node[vertex, fill, label=180:$1$] (1) at (5,5) {};
\node[vertex, fill, label=180:$2$] (2) at (4,4) {};
\node[vertex] (3) at (5,4) {};
\node[vertex] (4) at (6,4) {};
\node[vertex, fill, label=180:$4$] (5) at (4,3) {};
\node[vertex, fill, label=180:$3$] (6) at (5,3) {};
\node[vertex] (7) at (4,2) {};
\node[vertex] (8) at (5,2) {};

\draw[normaledge] (1) -- (2);
\draw[normaledge] (1) -- (3);
\draw[normaledge] (1) -- (4);
\draw[normaledge] (2) -- (5);
\draw[normaledge] (2) -- (6);
\draw[normaledge] (3) -- (5);
\draw[normaledge] (4) -- (6);
\draw[normaledge] (5) -- (7);
\draw[normaledge] (5) -- (8);
\draw[normaledge] (6) -- (8);

\begin{scope}[xshift=8cm]
\node[vertex, fill, label=180:$1$] (1) at (5,5) {};
\node[vertex, fill, label=180:$2$] (2) at (4,4) {};
\node[vertex, label=180:$3$] (3) at (5,4) {};
\node[vertex, label=0:$4$] (4) at (6,4) {};
\node[vertex, fill, label=180:$5$] (5) at (4,3) {};
\node[vertex, label=0:$6$] (6) at (5,3) {};
\node[vertex, label=180:$8$] (7) at (4,2) {};
\node[vertex, label=0:$7$] (8) at (5,2) {};

\draw[treeedge] (1) -- (2);
\draw[treeedge] (1) -- (3);
\draw[treeedge] (1) -- (4);
\draw[treeedge] (2) -- (5);
\draw[treeedge] (2) -- (6);
\draw[normaledge] (3) -- (5);
\draw[normaledge] (4) -- (6);
\draw[treeedge] (5) -- (7);
\draw[treeedge] (5) -- (8);
\draw[normaledge] (6) -- (8);
\end{scope}

\begin{scope}[xshift=4cm]
\node[vertex, fill, label=180:$1$] (1) at (5,5) {};
\node[vertex, fill, label=180:$2$] (2) at (4,4) {};
\node[vertex, label=180:$4$] (3) at (5,4) {};
\node[vertex, label=0:$3$] (4) at (6,4) {};
\node[vertex, fill, label=180:$6$] (5) at (4,3) {};
\node[vertex, fill, label=0:$5$] (6) at (5,3) {};
\node[vertex, label=180:$8$] (7) at (4,2) {};
\node[vertex, label=0:$7$] (8) at (5,2) {};

\draw[treeedge] (1) -- (2);
\draw[treeedge] (1) -- (3);
\draw[treeedge] (1) -- (4);
\draw[treeedge] (2) -- (5);
\draw[treeedge] (2) -- (6);
\draw[normaledge] (3) -- (5);
\draw[normaledge] (4) -- (6);
\draw[treeedge] (5) -- (7);
\draw[normaledge] (5) -- (8);
\draw[treeedge] (6) -- (8);
\end{scope}
\end{tikzpicture}
    \caption{Example that shows that the idea of \cref{lemma:bfs-xp} is not valid for LBFS. The left figure shows the input graph. The filled vertices are the elements of $S$ and the given numbers describe the fixed ordering of $S$. The graph in the middle shows an LBFS ordering whose \cf-tree fulfills the condition on its internal vertices. However, the right graph shows an LBFS ordering that could be computed by LBFS$^+$ but has less internal vertices. In particular, vertices $3$ and $4$ of the left graph have different children in the trees of the middle and the right graph.}
    \label{fig:counterexample-lbfs}
\end{figure}

\section{Conclusion}\label{sec:conclusion}

We have studied the Maximum (Minimum) Leaf Spanning Tree problem for various search algorithms. The achieved results can be contrasted with the results of Bergougnoux et al.~\cite{bergougnoux2025parameterized}, who studied the same problems for DFS trees. The complexity is effectively inverted, i.e., the problems for DFS are \W-hard when parameterized by the number of leaves, but in \FPT{} for GS and BFS-type searches and vice versa. This coincides with the intuitive assumption that a DFS will commonly lead to few leaves while a BFS yields many leaves.

\Cref{tab:results-tree} shows some remaining questions: We are rather sure that we can adapt the para-\NP-hardness proof of \MinInF{} to other clique starters, such as Maximum Cardinality Search (MCS) or Maximal Neighborhood Search (MNS). We also expect that \MaxInF{} is para-\NP-hard for LBFS for larger values of~$k$ (the case of $k=3$ can be solved by enumeration).

All of the results shown here are for first-in trees only. It is possible to state the same problems for last-in trees as well. This problem is addressed for a variety of searches in a separate article which is in preparation.

\bibliography{many-leaves}

\begin{thebibliography}{10}

\bibitem{aazami2008hardness}
Ashkan Aazami.
\newblock {\em Hardness results and approximation algorithms for some problems on graphs}.
\newblock PhD thesis, University of Waterloo, 2008.
\newblock URL: \url{http://hdl.handle.net/10012/4147}.

\bibitem{barioli2013parameters}
Francesco Barioli, Wayne Barrett, Shaun~M. Fallat, H.~Tracy Hall, Leslie Hogben, Bryan~L. Shader, Pauline van~den Driessche, and Hein van~der Holst.
\newblock Parameters related to tree-width, zero forcing, and maximum nullity of a graph.
\newblock {\em J. Graph Theory}, 72(2):146--177, 2013.
\newblock \href {https://doi.org/10.1002/JGT.21637} {\path{doi:10.1002/JGT.21637}}.

\bibitem{beisegel2019end}
Jesse Beisegel, Carolin Denkert, Ekkehard K{\"o}hler, Matja{\v{z}} Krnc, Nevena Piva{\v{c}}, Robert Scheffler, and Martin Strehler.
\newblock On the end-vertex problem of graph searches.
\newblock {\em Discrete Math. Theor. Comput. Sci.}, 21, 2019.
\newblock \href {https://doi.org/10.23638/DMTCS-21-1-13} {\path{doi:10.23638/DMTCS-21-1-13}}.

\bibitem{beisegel2021recognition}
Jesse Beisegel, Carolin Denkert, Ekkehard K{\"{o}}hler, Matja{\v z} Krnc, Nevena Piva{\v c}, Robert Scheffler, and Martin Strehler.
\newblock The recognition problem of graph search trees.
\newblock {\em {SIAM} J. Discrete Math.}, 35(2):1418--1446, 2021.
\newblock \href {https://doi.org/10.1137/20M1313301} {\path{doi:10.1137/20M1313301}}.

\bibitem{beisegel2024graph}
Jesse Beisegel, Ekkehard K{\"{o}}hler, Fabienne Ratajczak, Robert Scheffler, and Martin Strehler.
\newblock Graph search trees and the intermezzo problem.
\newblock In {\em {MFCS} 2024}, volume 306 of {\em LIPIcs}, pages 22:1--22:18, 2024.
\newblock \href {https://doi.org/10.4230/LIPICS.MFCS.2024.22} {\path{doi:10.4230/LIPICS.MFCS.2024.22}}.

\bibitem{bergougnoux2025parameterized}
Benjamin Bergougnoux, Nello Blaser, Michael~R. Fellows, Petr~A. Golovach, Frances~A. Rosamond, and Emmanuel Sam.
\newblock On the parameterized complexity of lineal topologies (depth-first spanning trees) with many or few leaves.
\newblock {\em J. Comput. Syst. Sci}, 154:103680, 2025.
\newblock \href {https://doi.org/10.1016/J.JCSS.2025.103680} {\path{doi:10.1016/J.JCSS.2025.103680}}.

\bibitem{bhyravarapu2025parameterized}
Sriram Bhyravarapu, Lawqueen Kanesh, Madhumita Kundu, Daniel Lokshtanov, and Saket Saurabh.
\newblock Parameterized algorithms for power edge set and zero forcing set.
\newblock In {\em {IWOCA} 2025}, pages 349--361, 2025.
\newblock \href {https://doi.org/10.1007/978-3-031-98740-3_25} {\path{doi:10.1007/978-3-031-98740-3_25}}.

\bibitem{bodlaender1993linear}
Hans~L. Bodlaender.
\newblock On linear time minor tests with depth-first search.
\newblock {\em J. Algorithms}, 14(1):1--23, 1993.
\newblock \href {https://doi.org/10.1006/JAGM.1993.1001} {\path{doi:10.1006/JAGM.1993.1001}}.

\bibitem{bonsma2008spanning}
Paul~S. Bonsma and Florian Zickfeld.
\newblock Spanning trees with many leaves in graphs without diamonds and blossoms.
\newblock In {\em {LATIN} 2008}, volume 4957 of {\em LNCS}, pages 531--543, 2008.
\newblock \href {https://doi.org/10.1007/978-3-540-78773-0_46} {\path{doi:10.1007/978-3-540-78773-0_46}}.

\bibitem{bresar2017grundy}
Bo\v{s}tjan Bre\v{s}ar, Csilla Bujt{\'{a}}s, Tanja Gologranc, Sandi Klav\v{z}ar, Ga\v{s}per Ko\v{s}mrlj, Bal{\'{a}}zs Patk{\'{o}}s, Zsolt Tuza, and M{\'{a}}t{\'{e}} Vizer.
\newblock Grundy dominating sequences and zero forcing sets.
\newblock {\em Discrete Optim.}, 26:66--77, 2017.
\newblock \href {https://doi.org/10.1016/J.DISOPT.2017.07.001} {\path{doi:10.1016/J.DISOPT.2017.07.001}}.

\bibitem{bresar2014dominating}
Bo\v{s}tjan Bre\v{s}ar, Tanja Gologranc, Martin Milani\v{c}, Douglas~F. Rall, and Romeo Rizzi.
\newblock Dominating sequences in graphs.
\newblock {\em Discrete Math.}, 336:22--36, 2014.
\newblock \href {https://doi.org/10.1016/J.DISC.2014.07.016} {\path{doi:10.1016/J.DISC.2014.07.016}}.

\bibitem{bresar2016total}
Bo\v{s}tjan Bre\v{s}ar, Michael~A. Henning, and Douglas~F. Rall.
\newblock Total dominating sequences in graphs.
\newblock {\em Discrete Math.}, 339(6):1665--1676, 2016.
\newblock \href {https://doi.org/10.1016/J.DISC.2016.01.017} {\path{doi:10.1016/J.DISC.2016.01.017}}.

\bibitem{charbit2014influence}
Pierre Charbit, Michel Habib, and Antoine Mamcarz.
\newblock Influence of the tie-break rule on the end-vertex problem.
\newblock {\em Discrete Math. Theor. Comput. Sci.}, 16(2):57--72, 2014.
\newblock \href {https://doi.org/10.46298/DMTCS.2081} {\path{doi:10.46298/DMTCS.2081}}.

\bibitem{corneil2016tie}
Derek~G. Corneil, J{\'{e}}r{\'{e}}mie Dusart, Michel Habib, Antoine Mamcarz, and Fabien de~Montgolfier.
\newblock A tie-break model for graph search.
\newblock {\em Discrete Appl. Math.}, 199:89--100, 2016.
\newblock \href {https://doi.org/10.1016/J.DAM.2015.06.011} {\path{doi:10.1016/J.DAM.2015.06.011}}.

\bibitem{corneil2010end}
Derek~G. Corneil, Ekkehard K{\"o}hler, and Jean-Marc Lanlignel.
\newblock On end-vertices of {L}exicographic {B}readth {F}irst {S}earches.
\newblock {\em Discrete Appl. Math.}, 158(5):434--443, 2010.
\newblock \href {https://doi.org/10.1016/j.dam.2009.10.001} {\path{doi:10.1016/j.dam.2009.10.001}}.

\bibitem{corneil2008unified}
Derek~G. Corneil and Richard Krueger.
\newblock A unified view of graph searching.
\newblock {\em {SIAM} J. Discrete Math.}, 22(4):1259--1276, 2008.
\newblock \href {https://doi.org/10.1137/050623498} {\path{doi:10.1137/050623498}}.

\bibitem{cygan2015param}
Marek Cygan, Fedor~V. Fomin, {\L}ukasz Kowalik, Daniel Lokshtanov, D{\'a}niel Marx, Marcin Pilipczuk, Micha{\l} Pilipczuk, and Saket Saurabh.
\newblock {\em Parameterized Algorithms}.
\newblock Springer, 2015.
\newblock \href {https://doi.org/10.1007/978-3-319-21275-3} {\path{doi:10.1007/978-3-319-21275-3}}.

\bibitem{downey1995fixed}
Rodney~G. Downey and Michael~R. Fellows.
\newblock Fixed-parameter tractability and completeness~{I:} basic results.
\newblock {\em {SIAM} J. Comput.}, 24(4):873--921, 1995.
\newblock \href {https://doi.org/10.1137/S0097539792228228} {\path{doi:10.1137/S0097539792228228}}.

\bibitem{downey2013fundamentals}
Rodney~G. Downey and Michael~R. Fellows.
\newblock {\em Fundamentals of Parameterized Complexity}.
\newblock Texts in Computer Science. Springer, 2013.
\newblock \href {https://doi.org/10.1007/978-1-4471-5559-1} {\path{doi:10.1007/978-1-4471-5559-1}}.

\bibitem{fellows1992well}
Michael~R. Fellows and Michael~A. Langston.
\newblock On well-partial-order theory and its application to combinatorial problems of {VLSI} design.
\newblock {\em {SIAM} J. Discrete Math.}, 5(1):117--126, 1992.
\newblock \href {https://doi.org/10.1137/0405010} {\path{doi:10.1137/0405010}}.

\bibitem{fellows2000coordinatized}
Michael~R. Fellows, Catherine McCartin, Frances~A. Rosamond, and Ulrike Stege.
\newblock Coordinatized kernels and catalytic reductions: An improved {FPT} algorithm for max leaf spanning tree and other problems.
\newblock In {\em {FST} {TCS} 2000}, volume 1974 of {\em LNCS}, pages 240--251, 2000.
\newblock \href {https://doi.org/10.1007/3-540-44450-5_19} {\path{doi:10.1007/3-540-44450-5_19}}.

\bibitem{fomin2013linear}
Fedor~V. Fomin, Serge Gaspers, Saket Saurabh, and St{\'{e}}phan Thomass{\'{e}}.
\newblock A linear vertex kernel for maximum internal spanning tree.
\newblock {\em J. Comput. Syst. Sci}, 79(1):1--6, 2013.
\newblock \href {https://doi.org/10.1016/J.JCSS.2012.03.004} {\path{doi:10.1016/J.JCSS.2012.03.004}}.

\bibitem{fujie2003exact}
Tetsuya Fujie.
\newblock An exact algorithm for the maximum leaf spanning tree problem.
\newblock {\em Comput. Oper. Res.}, 30(13):1931--1944, 2003.
\newblock \href {https://doi.org/10.1016/S0305-0548(02)00117-X} {\path{doi:10.1016/S0305-0548(02)00117-X}}.

\bibitem{GareyJohnson}
Michael~R. Garey and David~S. Johnson.
\newblock {\em Computers and Intractability: A Guide to the Theory of NP-completeness}.
\newblock W. H. Freeman, 1979.

\bibitem{golumbic1980algorithmic}
Martin~Charles Golumbic.
\newblock {\em Algorithmic Graph Theory and Perfect Graphs}, volume~57 of {\em Annals of Discrete Mathematics}.
\newblock Elsevier, 2nd edition, 2004.
\newblock \href {https://doi.org/10.1016/S0167-5060(04)80052-9} {\path{doi:10.1016/S0167-5060(04)80052-9}}.

\bibitem{habib2000lexbfs}
Michel Habib, Ross~M. McConnell, Christophe Paul, and Laurent Viennot.
\newblock Lex-{BFS} and partition refinement, with applications to transitive orientation, interval graph recognition and consecutive ones testing.
\newblock {\em Theor. Comput. Sci.}, 234(1-2):59--84, 2000.
\newblock \href {https://doi.org/10.1016/S0304-3975(97)00241-7} {\path{doi:10.1016/S0304-3975(97)00241-7}}.

\bibitem{hagerup1985biconnected}
Torben Hagerup.
\newblock Biconnected graph assembly and recognition of {DFS} trees.
\newblock Technical Report A 85/03, Universit\"at des Saarlandes, 1985.
\newblock \href {https://doi.org/10.22028/D291-26437} {\path{doi:10.22028/D291-26437}}.

\bibitem{hagerup1985recognition}
Torben Hagerup and Manfred Nowak.
\newblock Recognition of spanning trees defined by graph searches.
\newblock Technical Report A 85/08, Universit\"at des Saarlandes, 1985.

\bibitem{hogben2022inverse}
Leslie Hogben, Jephian C.-H. Lin, and Bryan~L. Shader.
\newblock {\em Inverse Problems and Zero Forcing For Graphs}, volume 270 of {\em Mathematical Surveys and Monographs}.
\newblock AMS, 2022.
\newblock \href {https://doi.org/10.1090/surv/270} {\path{doi:10.1090/surv/270}}.

\bibitem{karp21}
Richard~M. Karp.
\newblock Reducibility among combinatorial problems.
\newblock In {\em Complexity of Computer Computations}, pages 85--103. Plenum Press, 1972.
\newblock \href {https://doi.org/10.1007/978-1-4684-2001-2_9} {\path{doi:10.1007/978-1-4684-2001-2_9}}.

\bibitem{Kleinberg06algorithmdesign}
Jon Kleinberg and \'Eva Tardos.
\newblock {\em Algorithm Design}.
\newblock Addison Wesley, 2006.

\bibitem{kneis2011new}
Joachim Kneis, Alexander Langer, and Peter Rossmanith.
\newblock A new algorithm for finding trees with many leaves.
\newblock {\em Algorithmica}, 61(4):882--897, 2011.
\newblock \href {https://doi.org/10.1007/S00453-010-9454-5} {\path{doi:10.1007/S00453-010-9454-5}}.

\bibitem{korach1989dfs}
Ephraim Korach and Zvi Ostfeld.
\newblock {DFS} tree construction: {A}lgorithms and characterizations.
\newblock In {\em {WG} 88}, volume 344 of {\em LNCS}, pages 87--106, 1989.
\newblock \href {https://doi.org/10.1007/3-540-50728-0_37} {\path{doi:10.1007/3-540-50728-0_37}}.

\bibitem{korhonen2021single}
Tuukka Korhonen.
\newblock A single-exponential time 2-approximation algorithm for treewidth.
\newblock In {\em {FOCS} 2021}, pages 184--192, 2021.
\newblock \href {https://doi.org/10.1109/FOCS52979.2021.00026} {\path{doi:10.1109/FOCS52979.2021.00026}}.

\bibitem{kratsch2015end}
Dieter Kratsch, Mathieu Liedloff, and Daniel Meister.
\newblock End-vertices of graph search algorithms.
\newblock In {\em {CIAC} 2015}, volume 9079 of {\em LNCS}, pages 300--312, 2015.
\newblock \href {https://doi.org/10.1007/978-3-319-18173-8_22} {\path{doi:10.1007/978-3-319-18173-8_22}}.

\bibitem{li2022simple}
Peng Li, Jianhui Shang, and Yi~Shi.
\newblock A simple linear time algorithm to solve the {MIST} problem on interval graphs.
\newblock {\em Theor. Comput. Sci.}, 930:77--85, 2022.
\newblock \href {https://doi.org/10.1016/J.TCS.2022.07.012} {\path{doi:10.1016/J.TCS.2022.07.012}}.

\bibitem{li2017deeper}
Wenjun Li, Yixin Cao, Jianer Chen, and Jianxin Wang.
\newblock Deeper local search for parameterized and approximation algorithms for maximum internal spanning tree.
\newblock {\em Inf. Comput.}, 252:187--200, 2017.
\newblock \href {https://doi.org/10.1016/J.IC.2016.11.003} {\path{doi:10.1016/J.IC.2016.11.003}}.

\bibitem{lin2019zero}
Jephian C.-H. Lin.
\newblock Zero forcing number, {G}rundy domination number, and their variants.
\newblock {\em Linear Algebra Appl.}, 563:240--254, 2019.
\newblock \href {https://doi.org/10.1016/j.laa.2018.11.003} {\path{doi:10.1016/j.laa.2018.11.003}}.

\bibitem{manber1990recognizing}
Udi Manber.
\newblock Recognizing breadth-first search trees in linear time.
\newblock {\em Inf. Process. Lett.}, 34(4):167--171, 1990.
\newblock \href {https://doi.org/10.1016/0020-0190(90)90155-Q} {\path{doi:10.1016/0020-0190(90)90155-Q}}.

\bibitem{prieto2003either}
Elena Prieto and Christian Sloper.
\newblock {Either/Or}: Using \textsc{Vertex Cover} structure in designing {FPT}-algorithms --- {T}he case of $k$-\textsc{Internal Spanning Tree}.
\newblock In {\em {WADS} 2003}, volume 2748 of {\em LNCS}, pages 474--483, 2003.
\newblock \href {https://doi.org/10.1007/978-3-540-45078-8_41} {\path{doi:10.1007/978-3-540-45078-8_41}}.

\bibitem{raible2010amortized}
Daniel Raible and Henning Fernau.
\newblock An amortized search tree analysis for $k$-leaf spanning tree.
\newblock In {\em {SOFSEM} 2010}, volume 5901 of {\em LNCS}, pages 672--684, 2010.
\newblock \href {https://doi.org/10.1007/978-3-642-11266-9_56} {\path{doi:10.1007/978-3-642-11266-9_56}}.

\bibitem{reich16complexity}
Alexander Reich.
\newblock Complexity of the maximum leaf spanning tree problem on planar and regular graphs.
\newblock {\em Theor. Comput. Sci.}, 626:134--143, 2016.
\newblock \href {https://doi.org/10.1016/J.TCS.2016.02.011} {\path{doi:10.1016/J.TCS.2016.02.011}}.

\bibitem{rong2022graph}
Guozhen Rong, Yixin Cao, Jianxin Wang, and Zhifeng Wang.
\newblock Graph searches and their end vertices.
\newblock {\em Algorithmica}, 84(9):2642--2666, 2022.
\newblock \href {https://doi.org/10.1007/s00453-022-00981-5} {\path{doi:10.1007/s00453-022-00981-5}}.

\bibitem{rong2026linear}
Guozhen Rong, Biao Yuan, Yongjie Yang, and Zhen Zhang.
\newblock A linear-time algorithm for the {MCS} end-vertex problem on chordal graphs: {A} bonus-driven search strategy.
\newblock In {\em {SOSA} 2026}, pages 298--311, 2026.
\newblock \href {https://doi.org/10.1137/1.9781611978964.23} {\path{doi:10.1137/1.9781611978964.23}}.

\bibitem{rosetarjanlueker76}
Donald~J. Rose, R.~Endre Tarjan, and George~S. Lueker.
\newblock Algorithmic aspects of vertex elimination on graphs.
\newblock {\em {SIAM} J. Comput.}, 5(2):266--283, 1976.
\newblock \href {https://doi.org/10.1137/0205021} {\path{doi:10.1137/0205021}}.

\bibitem{scheffler2025parameterized}
Robert Scheffler.
\newblock On the parameterized complexity of {G}rundy domination and zero forcing problems, 2025.
\newblock \href {https://arxiv.org/abs/2508.18104} {\path{arXiv:2508.18104}}.

\bibitem{scheffler2025partial}
Robert Scheffler.
\newblock The partial search order problem.
\newblock {\em Electron. J. Comb.}, 32(4), 2025.
\newblock \href {https://doi.org/10.37236/12654} {\path{doi:10.37236/12654}}.

\bibitem{scheffler2025leaves}
Robert Scheffler.
\newblock On the leaves of graph search trees.
\newblock {\em Ars Math. Contemp.}, 26, 2026.
\newblock \href {https://doi.org/10.26493/1855-3974.3238.2a8} {\path{doi:10.26493/1855-3974.3238.2a8}}.

\bibitem{sharma2022algorithms}
Gopika Sharma, Arti Pandey, and Michael~C. Wigal.
\newblock Algorithms for maximum internal spanning tree problem for some graph classes.
\newblock {\em J. Comb. Optim.}, 44(5):3419--3445, 2022.
\newblock \href {https://doi.org/10.1007/S10878-022-00897-4} {\path{doi:10.1007/S10878-022-00897-4}}.

\bibitem{spinrad1995algorithms}
Jeremy Spinrad and R.~Sritharan.
\newblock Algorithms for weakly triangulated graphs.
\newblock {\em Discrete Appl. Math.}, 59(2):181--191, 1995.
\newblock \href {https://doi.org/10.1016/0166-218X(93)E0161-Q} {\path{doi:10.1016/0166-218X(93)E0161-Q}}.

\bibitem{tarjan1984linear}
Robert~E. Tarjan and Mihalis Yannakakis.
\newblock Simple linear-time algorithms to test chordality of graphs, test acyclicity of hypergraphs, and selectively reduce acyclic hypergraphs.
\newblock {\em {SIAM} J. Comput.}, 13(3):566--579, 1984.
\newblock \href {https://doi.org/10.1137/0213035} {\path{doi:10.1137/0213035}}.

\bibitem{west2001introduction}
Douglas~B. West.
\newblock {\em Introduction to Graph Theory}.
\newblock Prentice Hall, 2001.

\bibitem{zehavi2018k-leaf}
Meirav Zehavi.
\newblock The $k$-leaf spanning tree problem admits a klam value of~39.
\newblock {\em Eur. J. Comb.}, 68:175--203, 2018.
\newblock \href {https://doi.org/10.1016/J.EJC.2017.07.018} {\path{doi:10.1016/J.EJC.2017.07.018}}.

\end{thebibliography}

\newpage

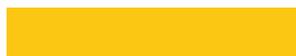
\appendix

\section[Proof of Theorem 3.11]{Proof of \cref{thm:tw}}\label{appendix}
The proof mainly rephrases the algorithm and proofs from Section~4.3 of~\cite{scheffler2025parameterized} and is adapted appropriately. We include the full proof here to facilitate the understanding instead of just pointing to \cite{scheffler2025parameterized} and highlighting the changes.

The algorithm uses an extended version of \emph{nice} tree-decompositions. In these decomposition the tree is considered to be a rooted binary tree. Furthermore, there are particular types of nodes. We will use a tree decomposition consisting of the following five node types.

\begin{description}
    \item[Leaf] Such a node is a leaf of $T$. Its bag is empty. 
    \item[Introduce Node] Such a node $t$ has exactly one child $t'$. There is a vertex $v \notin X_{t'}$ such that $X_{t} = X_{t'} \cup \{v\}$. We say that $v$ \emph{is introduced} at $t$. 
    \item[Forget Node] Such a node $t$ has exactly one child $t'$. There is a vertex $v \in X_{t'}$ such that $X_{t} = X_{t'} \setminus \{v\}$. We say that $v$ \emph{is forgotten} at $t$. 
    \item[Rule Node] Such a node $t$ has exactly one child $t'$ and it holds that $X_t = X_{t'}$. Furthermore, the parent of $t$ is a forget node. These tree nodes will be the place of our algorithm where we decide which particular rules are applied to the forgotten vertex.
    \item[Join Node] Such a node has exactly two children $t_1$ and $t_2$ and it holds that $X_t = X_{t_1} = X_{t_2}$.
\end{description}

Note that we also assume that the bag of the root node of the tree is empty, i.e., for every vertex of the graph there is exactly one forget node where this node is forgotten. For every $t \in V(T)$, we write $N_t(v)$ as short form of $N(v) \cap X_t$. Furthermore, for every $t \in V(T)$, we define $V^t := \{ v \mid v \in X_{t'} \text{ for some descendant $t'$ of $t$}\}$. As usual, our algorithm consists of a dynamic programming approach. For every node/bag of the tree decomposition, we will consider certain signatures that describe which decision we have made for the vertices within the bag. These signatures are six-tuples of the form $\Omega = (\Gamma, \Phi, b_\Gamma, b_\Phi, \D, \omega)$ whose entries are explained in the following.

\begin{description}
    \item[$\Gamma$-Type] For every vertex $v$, we have a value $\Gamma(v) \in \{\Z,\bot\}$ describing whether $v$ becomes blue via the $\Z$-rule ($\Gamma(v) = Z$) or whether it is part of the \Z-forcing set ($\Gamma(v) = \bot$).
    \item[$\Phi$-Type] For every vertex $v$, we have a value $\Phi(v) \in (\Z, E, \bot\}$ describing whether $v$ colors some other vertex blue ($\Phi(v) = \Z$), is the last vertex that is colored blue ($\Phi(v) = E$) or none of the two ($\Phi(v) = \bot$). 
    \item[Function $b_\Gamma$] For every vertex $v$, we have a boolean value $b_\Gamma(v)$. This value is $\true$ if and only if $\Gamma(v) = \bot$ or the vertex that colors $v$ blue was already fixed.
    \item[Function $b_\Phi$] For every vertex $v$, we have a boolean value $b_\Phi(v)$. This value is $\true$ if and only if $\Phi(v) \in \{E,\bot\}$ or the vertex $x$ that is colored blue by $v$ was already fixed.
    \item[Function $b_\Pi$] For every vertex $v$, we have a boolean value $b_\Pi(v)$. This value is $\true$ if and only if $\Gamma(v) = \bot$ or $\Phi(v) = E$ or the neighbor $x$ of $v$ that is white when $v$ becomes blue was already fixed.
    \item[Boolean Value $\lambda$] This value is true if and only if some vertex has already got the $\Phi$-value~$E$.
    \item[Dependency Graph \D] The dependency graph $\mathfrak{D}$ is a directed graph. For every vertex $v \in X_t$, the digraph $\D$ contains an \emph{event node} $\gamma_v$ and if $\Phi(v) \neq \bot$, it also contains an event node $\phi_v$. The vertex $\gamma_v$ represents the event when $v$ is made blue. The vertex $\phi_v$ represents the event when the $\Z$-rule is applied to $v$. An arc from one of these event nodes $x$ to another event node $y$ implies that the event represented by $x$ must happen before the event represented by $y$.
    \item[Weight $\omega$] Every signature is associated with a value $\omega \in \{0, \dots, n\} \cup \infty$ describing how many vertices in $V^t \setminus X_t$ have been assigned the $\Gamma$-type $\bot$, or -- in case that $\omega = \infty$ -- describing that $\D$ is not acyclic. We call signatures with $\omega < \infty$ \emph{valid}.
\end{description}

The main ideas of the algorithm are the following:

\begin{itemize}
    \item When a vertex $v$ is introduced in its introduce node, then we fix the $\Gamma$- and $\Phi$-type of $v$, i.e., we specify which rule type should make $v$ blue and whether some rule is applied to $v$ to make some other vertex blue. For both decisions we do not yet specify the other vertex of the corresponding rule as this vertex might not be part of the bag. Nevertheless, fixing the type of rule already implies some arcs for the dependency graph $\D$.
    \item We lazily fix the vertices of the rules in the rule nodes. This means, we only specify the other vertex $w$ of a rule concerning $v$ if either $v$ or $w$ is forgotten in the next bag. Again, fixing these vertices implies several arcs of the dependency graph~$\D$. These arcs also concern vertices that are neither $v$ nor $w$. For example, if we want to apply rule $v \z w$, then all other neighbors of $v$ in the bag must be blue before the $Z$-rule is applied to $v$. Conversely, assume $x$ is a neighbor of $v$ and it is fixed that some $T$-rule is applied to $x$ that colors some vertex different from $v$ blue. Then $v$ must become blue before the $T$-rule is applied to $x$. 
\end{itemize}

In the following, we will describe the detailed procedures for every possible type of tree node~$t$. In this descriptions, we will say that we \emph{bypass $v$ in $\D$} if we consider the subgraph $\D_v$ of $\D$ induced by the vertices $\gamma_v$ and $\phi_v$ as well as the union of their in-neighborhoods and out-neighborhoods and add all edges of the transitive closure of $\D_v$ to~$\D$.

\begin{description}
    \item[Leaf] There is one signature with $\omega=0$, empty functions $\Gamma$, $\Phi$, $b_\Gamma$, $b_\Phi$, and empty digraph~$\D$.
    \item[Introduce Node] Let $v$ be the introduced vertex and let $t'$ be the child of $t$. For every valid signature $\Omega$ of $t'$, we consider every possible choice $(\Gamma(v), \Phi(v))$ from $\{(\Z,\Z), (\Z, \bot),\allowbreak (\bot, \Z), (\bot, \bot)\}$ and -- if $\lambda(\Omega) = \false$ -- also $(\Z, E)$. For every of these choices we create a signature of $t$ by keeping the weight $\omega$ and adding vertex $\gamma_v$ and -- if $\Phi(v) = \Z$ -- vertex $\phi_v$ to $\D$. 
    
    If $\Gamma(v) = \bot$, then we set $b_\Gamma = b_\Pi = \true$. If $\Phi(v) \in \{\bot,E\}$, then we set $b_\Phi = \true$. If $\Phi(v) = E$, then we set $b_\Pi(v) = \true$. If $\Phi(v) = \Z$, then we add the arc $(\gamma_v, \phi_v)$ to $\D$.
    \item[Forget Node] Let $v$ be the forgotten vertex and let $t'$ be the child of $t$. For every valid signature of $t'$ where $b_\Gamma(v) = \true$, $b_\Phi(v) = \true$, and $b_\Pi(v) = \true$, we construct a signature of $t$ as follows. We remove $v$ from the functions $\Gamma$, $\Phi$, $b_\Gamma$ and $b_\Phi$.  We construct the new dependency graph by bypassing $v$ in $\D$ and deleting $\gamma_v$ and $\phi_v$ afterwards. If $\Gamma(v) = \bot$, then we increase the weight $\omega$ by one.
    \item[Join Node] Let $t_1$ and $t_2$ be the children of $t$. We say that $\Omega^1$ of $t_1$ and $\Omega^2$ of $t_2$ with $\Omega^i = (\Gamma^i, \Phi^i, b_\Gamma^i, b_\Phi^i, b_\Pi^i, \lambda^i, \D^i, \omega^i)$ are \emph{compatible} if 
    \begin{itemize}
        \item $\Gamma^1 = \Gamma^2$, $\Phi^1 = \Phi^2$,
        \item for every vertex $v$ with $\Gamma(v) \neq \bot$ it holds that $b^1_\Gamma(v) \land b^2_\Gamma(v) = \false$
        \item for every vertex $v$ with $\Phi(v) \notin \{\bot,E\}$ it holds that $b^1_\Phi(v) \land b^2_\Phi(v) = \false$
        \item for every vertex $v$ with $\Gamma(v) \neq \bot$ and $\Phi(v) \notin \{\bot,E\}$ it holds that $b^1_\Pi(v) \land b^2_\Pi(v) = \false$ 
        \item $\lambda^1 \land \lambda^2 = \false$.
        \item the union of $\D^1$ and $\D^2$ is acyclic. 
    \end{itemize}
    For every of those pairs of signatures, we create a signature of $t$ by keeping $\Gamma^1$ and $\Phi^1$, setting $b_\Gamma(v) = b^1_\Gamma(v) \lor b^2_\Gamma(v)$, $b_\Phi(v) = b^1_\Phi(v) \lor b^2_\Phi(v)$, $b_\Pi(v) = b^1_\Pi(v) \lor b^2_\Pi(v)$, $\lambda = \lambda^1 \lor \lambda^2$, $\D = \D_1 \cup \D_2$, and $\omega = \omega_1 + \omega_2$.
    \item[Rule Node] Let $t'$ be the child of $t$ and let $v$ be the vertex that is forgotten in the parent of~$t$. For every valid signature $\Omega$ of $t'$, we create signatures of $t$ in the following four-step procedure:
    \begin{enumerate}[(R1)]
        \item If $b_\Gamma(v) = \false$, then create for every $f \in N_t(v)$ with $\Phi(f) = \Gamma(v)$ and $b_\Phi(f) = \false$ a new signature $\Omega(f)$ which is a copy of $\Omega$. Vertex $f$ will be the vertex that colors $v$ blue. If $b_\Gamma(v) = \true$, i.e., the vertex coloring $v$ blue is already fixed, then create only one such copy $\Omega(\varnothing)$. \label{r1}
        \item For every of the signatures $\Omega(f)$ created in Step~(R\ref{r1}), do the following: If $b_\Phi(v) = \false$, then create for every $g \in N_t(v)$ with $\Gamma(g) = \Phi(v)$ and $b_\Gamma(g) = \false$, a new signature $\Omega(f,g)$ which is a copy of $\Omega(f)$. Vertex $g$ will be the vertex that is colored blue by $v$. If $b_\Phi(v) = \true$, i.e., the vertex colored blue by $v$ has already been fixed, then create only one such copy $\Omega(f,\varnothing)$.\label{r2}
        \item For every of the signatures $\Omega(f,g)$ created in Step~(R\ref{r2}), do the following: If $b_\Pi(v) = \false$, then create for every $h \in N_t(v)$ with $\Gamma(h) = \Z$ a new signature $\Omega(f,g,h)$ which is a copy of $\Omega(f,g)$. Vertex $h$ will be the vertex that is still white when $v$ is colored blue. If $b_\Pi(v) = \true$, then create only one such copy $\Omega(f,g,\varnothing)$.\label{r2b}
        \item If $\Gamma(v) = \Z$, then for every of the signatures $\Omega(f,g,h)$ created in Step~(R\ref{r2b}) and every set $W \subseteq X_t$ such that $b_\Pi(w) = \false$ for all $w \in S$, do the following: Create a new signature $\Omega(f,g,h,S)$ which is a copy of $\Omega(f,g,h)$. Vertex $v$ will be the vertex that is still white when the vertices of $W$ are colored blue. If $\Gamma(v) = \bot$, then create only one such copy $\Omega(f,g,h,\emptyset)$.\label{r2c}
        \item For every of the created signatures $\Omega(f,g,h,W)$, do the following:\label{r3}
         \begin{enumerate}
        \item If $f \neq \varnothing$, then add $(\phi_f, \gamma_v)$ to $\D$; set $b_\Gamma(v) = b_\Phi(f) = \true$. This means that we apply rule $f \z v$.\label{r3a}
        \item If $g \neq \varnothing$, then add $(\phi_v, \gamma_g)$ to $\D$; set $b_\Gamma(g) = b_\Phi(v) = \true$. This means that we apply rule $v \z g$.\label{r3b}
        \item If $h \neq \varnothing$, then add $(\gamma_v, \gamma_h)$ to $\D$; set $b_\Pi(v) = \true$. This means that $h$ is white when $v$ is colored blue.\label{r3b2}
         \item For each $w \in W$, add $(\gamma_w, \gamma_v)$ to $\D$; set $b_\Pi(w) = \true$. This means that $v$ is white when $w$ is colored blue.\label{r3b3}
        \item For each $w \in N_t(v) \setminus\{f\}$ with $\Phi(w) = \Z$, add arc $(\gamma_v,\phi_w)$ to~$\D$. If $\Phi(v) = E$, then also add arc $(\phi_w, \gamma_v)$.\label{r3c}
        \item If $\Phi(v) = \Z$, then for each $w \in N_t(v)$ with $w \neq g$, add arc $(\gamma_w,\phi_v)$ to $\D$. If $\Phi(w) = E$, then also add arc $(\phi_v,\gamma_w)$.\label{r3d}
        \item If $\Phi(v) = E$, then for each $w \in X_t \setminus \{v\}$, add arc $(\gamma_w,\gamma_v)$ to $\D$.\label{r3e}
        \item If there is some $w \in X_t \setminus \{v\}$ with $\Phi(w) = E$, then add arc $(\gamma_v,\gamma_w)$ to $\D$.\label{r3f}
        \end{enumerate}
    \item  Check for every of the signatures $\Omega(f,g)$ whether its dependency graph $\D$ is acyclic. If not, then set its weight $\omega$ to $\infty$.
    \end{enumerate}
\end{description}

Note that if we construct two signatures of the same node that only differ in the weight~$\omega$, then we only keep the signature with smaller weight.

In the following, we will prove that there is a valid signature of the root node with weight $\omega = k$ if and only if there is an $\R$-forcing set of $G$ of size $k$. First assume that there is a valid signature $\Omega_{\tilde{t}}$ of the root node ${\tilde{t}}$. We define the \emph{signature tree $\T$ of $\Omega_{\tilde{t}}$} to be the set of signatures that contains $\Omega_{\tilde{t}}$ and exactly one signature $\Omega_t$ for every $t \in V(T)$ such that the signature $\Omega_t \in \T$ was created using the signature $\Omega_{t'}$ where $t'$ is a child of $t$. We define $\fR$ to be the set of rules applied in~(R\ref{r3}) to compute the signatures of $\T$. Furthermore, we define $\D^*$ to be the union of the dependency graphs of all signatures in $\T$.

The following two lemmas present some properties of $\D^*$.

\begin{lemma}\label{lemma:rule-arcs}
If $p \z q$ is in $\fR$, then $(\gamma_p, \phi_p), (\phi_p, \gamma_q) \in \D^*$ and for all $r \in N(p)$ with $r \neq q$, it holds that $(\gamma_r, \phi_p) \in \D^*$. Furthermore, either there is a vertex $s \in N(q)$ with $(\gamma_q, \gamma_s) \in \D^*$ or $\gamma_q$ is the unique $\gamma$-vertex of some vertex $v$ with $\Gamma(v) \neq \bot$ with outdegree~0 in~$\D^*$.
\end{lemma}

\begin{proof}
    The first sentence follows by Statement~1 of Lemma~4.10 in \cite{scheffler2025parameterized}.

    For the second sentence, we distinguish two cases. First assume that $\Phi(q)$ has not been $E$ and in the signatures of the signature tree $\T$. Then $b_\Pi(q)$ was set to \false{} in the introduce node of $q$. In the forget node of $q$, the value $b_\Pi(q)$ has been \true{}, so there is some rule node where $b_\Pi(q)$ was set to \true. This happened in (R\ref{r3b2}) or in (R\ref{r3b3}). In that case, there has been some vertex $s$ such that $(\gamma_q,\gamma_s)$ was introduced to the dependency graph $\D$. Thus, this edge is also in $\D^*$. 
    
    It remains to show that if $\Phi(q) = E$, then the outdegree of $\gamma_w$ is zero. Assume for contradiction that $\Phi(q) = E$ but $\gamma_q$ has an outgoing arc in $\D^*$. Note that $b_\Pi(q) = \true$ in all signatures of bags that contain $q$. Then there must be such an arc that was not introduced when bypassing some other vertex. In the following we consider all ways how such an arc can be constructed. 
    
     First, it could be added in the introduce node of $q$. However, this only happens if $\Phi(q) = \Z \neq E$. So assume that it happens in the rule node $t$ of $q$. If this happens in (R\ref{r3b2}), then $h \neq \varnothing$ and, thus, $b_\Pi(q)$ has been $\false$ in (R\ref{r2b}), a contradiction to the observation above. if this happens in (R\ref{r3c}), then the reverse arc of the respective arc is also added to $\D$. So the signature would not be valid. It could happen in (R\ref{r3f}). However, then there would be another vertex with $\Phi(w) = E$. This is not possible as the value $\lambda$ ensures that such a vertex is introduced at most once. 
     
     The same arguments can be used if the arc is introduced in the rule node of another vertex. The uniqueness of the vertex $\gamma_q$ follows by the fact that at most one vertex of $\Phi$-type $E$ is introduced.
\end{proof}

\begin{lemma}\label{lemma:dependency-cycle}
   $\D^*$ is acyclic. 
\end{lemma}

\begin{proof}
    The proof works exactly like the proof of Lemma~4.11 in \cite{scheffler2025parameterized}.
\end{proof}

These lemmas enable us to prove the existence of an \R-forcing set if we have found a valid signature in the root node.

\begin{lemma}\label{lemma:tw-hin}
    If there is a valid signature $\Omega_{\tilde{t}}$ in the root node ${\tilde{t}}$ with weight $\omega = k$, then there is an $\Z$-forcing set of $G$ of size $k$.
\end{lemma}

\begin{proof}
    Consider the signatures in the signature tree $\T$. Note that a vertex of $G$ has the same $\Gamma$-type in all signatures of $\T$. Hence, there are exactly $k$ vertices of $G$ that got the $\Gamma$-type~$\bot$. Let $S$ be the set of those vertices. Since forget nodes only adopt those signatures from their child where the forgotten vertex has $b_\Gamma$-value $\true$, every vertex that is not in $S$ has been colored blue by some $\Z$-rule. Due to \cref{lemma:dependency-cycle}, the complete dependency graph $\D^*$ is acyclic. Thus, there is a topological sorting of $\D^*$. We claim that there is also a sorting fulfilling certain properties.

    \begin{claim}
        There is a topological sorting $\sigma$ of $\D^*$ such that $\sigma$ starts with the $\gamma$-nodes of the vertices in $S$ and if there is a rule $v \z w$ in $\fR$, then $\phi_v$ is the direct predecessor of $\gamma_w$ in the sorting. 
    \end{claim}

    \begin{claimproof}
        It is easy to check that a vertex $\gamma_x$ only gets an incoming edge if $\Gamma(x) = \Z$ or $\Phi(x) = E$. Since both properties do not hold for the vertices in $S$, their $\gamma$-nodes have in-degree zero and, thus, can be taken at the beginning of the topological sorting.

        We construct the graph $\D^{**}$ from $\D^*$ by adding for each $v \z w$ in $\fR$ and each arc $(x,\gamma_w) \in D^*$ that is not a bypassing arc also the arc $(x,\phi_v)$. We claim that $\D^{**}$ is still acyclic. Assume for contradiction that this would be not case. Then there must be a cycle $C$ in $\D^{**}$ that does not contain any bypassing arc as such arcs can be replaced by non-bypassing arcs.  Let $C$ be a cycle that contains the least number of arcs that were added to $\D^{**}$. As $\D^*$ is acyclic, $C$ must contain some new arc $(x,\phi_v)$, i.e., there have been the rule $v \z w$ in $\fR$ and the non-bypassing arc $(x,\gamma_w)$ in $\D^*$.  We observe that all non-bypassing arcs going into $\gamma_w$ have to come either from $\phi_v$ or are coming from some vertex $\gamma_u$. This holds because the only other possibility is the construction of an arc from some $\phi$-vertex to $\gamma_w$ in (R\ref{r3c}) or (R\ref{r3d}) but this would have closed a cycle in the dependency graph of some signature. This also implies that the only non-bypassing arc in $\D^{**}$ that leaves $\phi_v$ is $(\phi_v,\gamma_w)$. Therefore, $C$ must contain $(\phi_v,\gamma_w)$ and we can replace $(x,\phi_v)$ and $(\phi_v,\gamma_w)$ in $C$ by $(x,\gamma_w)$ and get a cycle that contains less of the newly added arcs; a contradiction.

        If we now construct a topological sorting of $\D^{**}$, then for every rule $v \z w$, we can take $\gamma_w$ directly after $\phi_v$, due to the added arcs. This topological sorting is also a topological sorting of $\D^*$
    \end{claimproof}

    So let $\sigma$ have the form as stated in the claim. We arrange the $\Z$-rules in $\fR$ according to the order of the respective event nodes in $\sigma$. Note that this is possible since the respective $\phi$- and $\gamma$-nodes appear consecutively in $\sigma$. We claim that this ordering of the rules is a valid ordering, i.e., all respective vertices have the right color. This follows from \cref{lemma:rule-arcs}. If we have a rule $p \z q$, then $p$ is already blue and $q$ is still white (due to $(\gamma_p,\phi_p), (\phi_p, \gamma_q) \in \D^*$). All other neighbors $r$ of $p$ are blue since $(\gamma_r, \phi_p) \in \D^*$. Furthermore, there is some white neighbor of $q$ (since $(\gamma,q,\gamma_s) \in D^*$ or $q$ is the last vertex to become blue. This finalizes the proof.
\end{proof}

It remains to show that an $\Z$-forcing set of size $\leq k$ can also be found by our algorithm.

\begin{lemma}\label{lemma:tw-rueck}
    If there is an $\Z$-forcing set of $G$ of size $\leq k$, then there is a valid signature $\Omega_{\tilde{t}}$ in the root node ${\tilde{t}}$ with weight $\omega \leq k$.
\end{lemma}

\begin{proof}
    Let $S$ be an $\Z$-forcing set of size $\leq k$ and let $\sigma = (R_1, \dots, R_\ell)$ be the sequence of $\Z$-rules applied to $G$. For every tree node, we only consider signatures that match $S$ and $\sigma$, i.e, if $v \in S$, then $\Gamma(v) = \bot$ and if there is a rule $v \z w$, then $\Phi(v) = \Z$ and $\Gamma(w) = \Z$. In every rule node we apply exactly the rule to the forgotten vertex that is part of $\sigma$. We claim that doing this, we construct a valid signature of the root node $\tilde{t}$ with weight $\omega \leq k$.

    To this end, we define the dependency graph $\D^\sigma$. Here, we add arcs with respect to $\sigma$, i.e., for every rule we add the arcs as given in \cref{lemma:rule-arcs}. Furthermore, if rule $p \z q$ is applied before rule $r \z s$, then we add the arc $(\gamma_q,\phi_r)$. Let $\D^\sigma$ be the transitive closure of the digraph containing all these arcs. It is clear that the graph $\D^\sigma$ is acyclic since it represents the linear order of the rules of $\sigma$. Furthermore, the graph $\D^*$ is a subgraph of $\D^\sigma$ and, thus, $\D^*$ is also acyclic. As the dependency graphs of every signature are subgraphs of $\D^*$, they are also acyclic and, thus, all these signatures are valid. Hence, the algorithm has computed a valid signature of the root node $\tilde{t}$. Since we only have used $|S|$ many vertices with $\Gamma$-value $\bot$, the weight of this signature has to be at most $|S| \leq k$.
\end{proof}

Finally, we can prove \cref{thm:tw}.

\begin{proof}[Proof of \cref{thm:tw}]
    \Cref{lemma:tw-hin,lemma:tw-rueck} show that the algorithm works correctly. So consider the running time. We apply Korhonen's algorithm~\cite{korhonen2021single} to compute a tree decomposition of width $\leq 2d +1$ in $2^{\Oc(d)} \cdot n$ time. It is easy to see that we can bring such a decomposition into the form using our five node types in $\Oc(n)$ time. Let $\ell = 2d + 2$. Then for a particular tree node $t$, there are at most $2^\ell \cdot 3^\ell \cdot 2^\ell \cdot 2^\ell \cdot 2^\ell \cdot 2 \cdot 2^{\Oc(\ell^2)} \in 2^{\Oc(d^2)}$ different choices for $(\Gamma, \Phi, b_\Gamma, b_\Phi, b_\Pi, \lambda, \D)$. For every of these choices, we have at most one signature of $t$. Since there are $\Oc(n)$ tree nodes, the total number of signatures is bounded by $2^{\Oc(d^2)} \cdot n$. The number of steps used for one of these signatures is polynomial in $d$. Overall, this implies a running time of $\Oc(d^c) \cdot 2^{\Oc(d^2)} \cdot n = 2^{\Oc(d^2)} \cdot n$.
\end{proof}

\end{document}